\RequirePackage{fix-cm}
\documentclass[smallextended]{svjour3}       %
\smartqed  %
\usepackage{graphicx}
\usepackage{smartdiagram}
\usepackage{graphicx}

\usepackage{cancel}

\usepackage{wrapfig}
\usepackage{booktabs}
\usepackage{tabularx}
\usepackage{subcaption}
\usepackage{tikz}
\usepackage{tkz-graph}
\usepackage{hyperref}

\usepackage{doi}

\usepackage{pgf}
\usepackage{pgf-pie}
\usetikzlibrary{arrows,automata}
\usepackage{amsmath} %
\usepackage{onimage}

\usepackage{amssymb}	%
\usepackage{xspace}

\usepackage{msymb}
\newcommand{\limply}{\rightarrow}
\usepackage{stmaryrd}

\usepackage{ntheorem}
\newcommand*{\TakeFourierOrnament}[1]{{%
\fontencoding{U}\fontfamily{futs}\selectfont\char#1}}
\newcommand*{\danger}{\TakeFourierOrnament{66}}
\newtheorem*{dangerous}{\Large\danger}

\usepackage[inline]{enumitem}

\usepackage{prettyref}
\newcommand{\rref}[2][]{\prettyref{#2}}
\newrefformat{sec}{Section\,\ref{#1}}
\newrefformat{subsec}{Section\,\ref{#1}}
\newrefformat{def}{Def.\,\ref{#1}}
\newrefformat{thm}{Theorem\,\ref{#1}}
\newrefformat{prop}{Proposition\,\ref{#1}}
\newrefformat{rem}{Remark\,\ref{#1}}
\newrefformat{lem}{Lemma\,\ref{#1}}
\newrefformat{cor}{Cor.\,\ref{#1}}
\newrefformat{ex}{Example\,\ref{#1}}
\newrefformat{tab}{Table\,\ref{#1}}
\newrefformat{fig}{Fig.\,\ref{#1}}
\newrefformat{case}{case\,\ref{#1}}
\newrefformat{foot}{Footnote\,\ref{#1}}
\newrefformat{app}{Appendix\,\ref{#1}}

\usesmartdiagramlibrary{additions}

\newcommand{\dL}{\text{\upshape\textsf{d{\kern-0.05em}L}}\xspace}

\newcommand{\KeYmaeraX}{KeYmaera~X\xspace}

\usepackage{float}
\usetikzlibrary{fadings,shapes.arrows,shadows}
\tikzfading[name=arrowfading, top color=transparent!0, bottom color=transparent!100]
\tikzset{arrowfill/.style={top color=red!10, bottom color=red, general shadow={fill=black, shadow yshift=-0.8ex, path fading=arrowfading}}}
\tikzset{arrowstyle/.style={draw=gray,arrowfill, single arrow,minimum height=#1, single arrow,
single arrow head extend=.4cm,}}

\newcommand{\polynomials}[2]{\ensuremath{#1[#2]}}

\definecolor{reggreen}{rgb}{.6328125,.92578125,.69921875}
\definecolor{regred}{rgb}{.98046875,.421875,.46484375}

\newcommand{\Init}{\mathit{Init}}
\newcommand{\Safe}{\mathit{Safe}}
\newcommand{\Unsafe}{\mathit{Unsafe}}

\newcommand*{\D}[2][]{#2'}

\definecolor{lsblue}{HTML}{16303A}
\definecolor{lslightblue}{HTML}{2E6579}
\definecolor{lsverylightblue}{HTML}{4699B9}
\definecolor{lsgreen}{HTML}{5ECEF9}
\definecolor{lslightgreen}{HTML}{54B9DF}
\definecolor{lsrealgreen}{HTML}{2FCC49}
\definecolor{lsred}{HTML}{CC2F4F}
\definecolor{lsochre}{HTML}{B26817}
\definecolor{lsyellow}{HTML}{FFD91C}

\usetikzlibrary{arrows,shapes,positioning,trees}

\tikzset{
  basic/.style  = {draw, text width=3cm, drop shadow, font=\sffamily, rectangle},
  root/.style   = {basic, rounded corners=2pt, thin, align=center,
                   fill=lsgreen!30},
  level 2/.style = {basic, rounded corners=6pt, thin,align=center, fill=lsrealgreen!60,
                   text width=10em},
  level 3/.style = {basic, thin, align=left, fill=lsred!30, text width=9em}
}

\usepackage{pgfplotstable}
\usetikzlibrary{plotmarks}
\usepgfplotslibrary{colormaps}
\usepgfplotslibrary{colorbrewer}

\pgfplotsset{
    mark repeat*/.style={
        scatter,
        scatter src=x,
        scatter/@pre marker code/.code={
            \pgfmathtruncatemacro\usemark{
                or(mod(\coordindex,#1)==0, (\coordindex==(\numcoords-1))
            }
            \ifnum\usemark=0
                \pgfplotsset{mark=none}
            \fi
        },
        scatter/@post marker code/.code={}
    },
    colormap/PuOr,
    cycle list/PuOr,
    cycle multiindex* list={
			[indices of colormap={1,2,3,9,8,7} of PuOr]\nextlist
			solid,thick\\solid,thick\\solid,thick\\solid,thick\\solid,thick\\dashed,thick\\dashed,thick\\dashed,thick\\dashed,thick\\\nextlist
			mark=*,mark=square*,mark=triangle*,mark=diamond*,
			mark=x,
			mark=*,mark=square*,mark=triangle*,mark=diamond*,
			mark=x\nextlist
		},
		mark repeat*=5
}

\begin{document}

\title{Pegasus: Sound Continuous Invariant Generation\thanks{%
This material is based upon work supported by the National Science Foundation under Award CNS-1739629 and under Graduate Research Fellowship Grants Nos. DGE1252522 and DGE1745016, by AFOSR under grant number FA9550-16-1-0288, by the United States Air Force and DARPA under Contract No. FA8750-18-C-0092, and by the Alexander von Humboldt Foundation. The third author was supported by A$^*$STAR, Singapore.
Any opinions, findings, and conclusions or recommendations expressed in this material are those of the authors and do not necessarily reflect the views of any sponsoring institution, the U.S. government or any other entity.
}%
}

\author{Andrew~Sogokon        \and
        Stefan~Mitsch     \and
        Yong~Kiam~Tan     \and
        Katherine~Cordwell \and
        Andr\'{e}~Platzer
}

\institute{  A. Sogokon$^*$, S. Mitsch, Y.K. Tan, K. Cordwell and A. Platzer \at
             Computer Science Department\\
             Carnegie Mellon University\\
             5000 Forbes Avenue\\
             Pittsburgh, PA 15213, USA\\
             \email{\texttt{\{asogokon|smitsch|yongkiat|kcordwel|aplatzer\}@cs.cmu.edu}}\\
	     \emph{$^*$}Now at ECS, University of Southampton, UK.  %
}

\date{Received: date / Accepted: date}

\maketitle

\begin{abstract}
\emph{Continuous invariants} are an important component in deductive verification of hybrid and continuous systems. Just like discrete invariants are used to reason about correctness in discrete systems without having to unroll their loops, continuous invariants are used to reason about differential equations without having to solve them.
\emph{Automatic generation} of continuous invariants remains one of the biggest practical challenges to the automation of formal proofs of safety for hybrid systems.
There are at present many disparate methods available for generating continuous invariants; however, this wealth of diverse techniques presents a number of challenges, with different methods having different strengths and weaknesses.
To address some of these challenges, we develop \emph{Pegasus}: an automatic continuous invariant generator which allows for combinations of various methods, and integrate it with the \KeYmaeraX theorem prover for hybrid systems.
We describe some of the architectural aspects of this integration, comment on its methods and challenges, and present an experimental evaluation on a suite of benchmarks.
\end{abstract}
\section{Introduction}
\label{sec:introduction}

Safety verification problems for ordinary differential equations (ODEs) are continuous analogs to Hoare triples: the objective is to show that an ODE cannot evolve out of a designated set of safe states from any of its designated initial states.
The role of continuous invariants is broadly analogous to that of inductive invariants for discrete program verification.
A continuous invariant is a set of states that can never be left when following the ODE from that set; such an invariant implies safety when it contains all of the initial states and is also a subset of the safe states.
The problem of automatically generating invariants (also known as \emph{invariant synthesis})
is one of the greatest practical challenges in deductive verification of both continuous and discrete
systems. In theory, it is actually the \emph{only} challenge for hybrid systems safety~\cite{DBLP:conf/lics/Platzer12b}.

The proliferation of published techniques~\cite{DBLP:conf/cdc/SassiGS14,DBLP:conf/hybrid/KongBSJH17,DBLP:conf/emsoft/LiuZZ11,DBLP:journals/fmsd/PlatzerC09,DBLP:journals/tcs/RebihaMM15,DBLP:conf/hybrid/Rodriguez-CarbonellT05,DBLP:conf/vmcai/SogokonGJP16,DBLP:conf/hybrid/Tiwari08,DBLP:conf/hybrid/TiwariK04} for continuous invariant generation---targeting various classes of systems, and having different strengths and weaknesses---presents a complication: ideally, one does not want to be restricted by the limitations of one particular generation technique (or small family of techniques). Instead, it is far more desirable to have a framework that accommodates existing generation methods, allows for their combination, and is extensible with new methods as they become available.
In this article we (partially) meet the above challenge by
developing a single framework which allows us to combine invariant generation methods into novel invariant generation \emph{strategies}. In our work, we are guided by the following considerations:
\begin{enumerate}
\item Specialized invariant generation methods are effective only when the problem falls within their domain; their use must therefore be targeted.
\item A combination of invariant generation methods can be more practical than
any of the methods considered in isolation. A flexible and reconfigurable mechanism for combining
these methods is thus highly desirable.
\item Reasoning with automatically generated invariants needs to be done in a \emph{sound} fashion: any deficiencies in the generation procedure must not compromise the final verification result.
\end{enumerate}
Our interest in automatic invariant generation is motivated by the pressing need
to enhance the level of proof automation in deductive verification tools for hybrid systems. In this work we target the \mbox{\KeYmaeraX} theorem prover~\cite{Fulton2015}.

\paragraph{Contributions.}
This article is an extended version of the conference
paper~\cite{DBLP:conf/fm/SogokonMTCP19}.  The article describes the design and
implementation of a continuous invariant generator (Pegasus)\footnote{%
{\bf An etymological note on naming conventions.} The
KeY~\cite{DBLP:conf/cade/BeckertGHKRSS07} prover provided the foundation for
developing KeYmaera~\cite{DBLP:conf/cade/PlatzerQ08}, an interactive theorem
prover for hybrid systems. The name KeYmaera was a pun on the
\textit{Chimaera}, a hybrid monster from Classical Greek mythology. The tactic
language of the new (aXiomatic) \KeYmaeraX prover~\cite{Fulton2015} is called
Bellerophon~\cite{Fulton2017}, after the hero who defeats the Chimaera in the
myth. In keeping with an established tradition, the invariant generation
framework is called Pegasus because the aid of this winged horse was crucial to
Bellerophon in his feat.}
 and its integration into \KeYmaeraX.
It outlines some of the principles behind this coupling, the techniques used to generate invariants,
and the mechanism used for combining them into more powerful invariant generation strategies.
An evaluation of this integration on a set of verification benchmarks is presented---with very
promising results. The present article extends our previous work~\cite{DBLP:conf/fm/SogokonMTCP19} with:

\begin{enumerate}
\item Extensive coverage of the \emph{methods} for generating continuous invariants employed by Pegasus (\rref{sec:pegasus}), including extended descriptions of several invariant generation methods, as well as new material on \emph{conic abstractions}~\cite{DBLP:conf/formats/BogomolovGHK17} and on the theory and practice of generating \emph{rational first integrals} for non-linear and linear systems~\cite{falconi2004,ferragut2010,gorbuzov2012,man1993,Man1994,Schlomiuk1993}. The extended article also includes a detailed account of the pitfalls and caveats associated with the various invariant generation and checking methods (Sections~\ref{sec:checking}--\ref{sec:evaluation}).
\item New insights on invariant generation \emph{strategies} based on combining various invariant generation methods (\rref{sec:saturation}), including various configuration options for the \emph{differential saturation}~\cite{DBLP:journals/fmsd/PlatzerC09} strategy and a new strategy based on \emph{differential divide-and-conquer}~\cite{DBLP:conf/vmcai/SogokonGJP16}.
\item An extended benchmark suite with $60$ new problems on top of the $90$ existing ones (\rref{sec:evaluation}), together with extended experimental evaluation and analysis of various invariant generation strategy configurations.
\end{enumerate}

\paragraph{Structure of this article.}
Mathematical preliminaries and definitions are reviewed in~\rref{sec:preliminaries}.
\rref{sec:checking} recalls the problem of continuous invariant \emph{checking} and describes our architecture for \emph{sound} invariant checking and generation.
Sections~\ref{sec:pegasus} and~\ref{sec:saturation} describe some of the methods employed by Pegasus for generating continuous invariants, along with mechanisms for their combination.
\rref{sec:evaluation} presents an empirical evaluation of our integration with \KeYmaeraX on a suite of verification benchmarks.
\rref{sec:related} reviews related work and \rref{sec:outlook} discusses the outlook and possible further extensions.
\rref{sec:conclusion} ends with a summary and concluding remarks.

\section{Preliminaries}
\label{sec:preliminaries}
\paragraph{Ordinary Differential Equations.}
An $n$-dimensional autonomous system of first-order ODEs has the form:
\(   \D{\vec{x}} = f(\vec{x}),  \)
where $\vec{x} = (x_1,\dots,x_n) \in \reals^n$ is a vector of state
variables, $\D{\vec{x}}=(\D{x}_1,\dots,\D{x}_n)$ denotes their
time-derivatives, i.e. $\frac{d x_i}{dt}$ for each
$i=1,\dots,n$, and $f(\vec{x}) = (f_1(\vec{x}), \dots, f_n(\vec{x}))$ specify the right-hand side (RHS) of the equations that these time-derivatives must obey along solutions to the ODEs.
Geometrically, such a system of ODEs defines a \emph{vector field}
$f:\mathbb{R}^n \to \mathbb{R}^n$, associating to each point $\vec{x} \in
\mathbb{R}^n$ the vector  $f(\vec{x}) = (f_1(\vec{x}),
\dots,f_n(\vec{x})) \in \mathbb{R}^n$ specifying in which direction the continuous system evolves at $\vec{x}$.
Whenever the state of the system is required to be confined within some prescribed set
of states $Q \subseteq \mathbb{R}^n$, called its \emph{evolution domain constraint}\footnote{%
Evolution domain constraints are also called \emph{mode invariants} in the context of hybrid automata. We avoid this name to prevent fundamental confusion with generated invariants.}, we will write $\D{\vec{x}} =
f(\vec{x})~\&~Q$. If no evolution domain constraint is specified, then $Q = \mathbb{R}^n$.
A \emph{solution} to the initial value problem for the system of ODEs
\mbox{$\D{\vec{x}} =f(\vec{x})$} with initial value $\vec{x}_0 \in
\mathbb{R}^n$ is a differentiable function  $\vec{x}(\vec{x}_0,t) :
(a,b) \to \mathbb{R}^n$ defined on some \emph{maximal interval of existence} $(a,b) \subseteq \mathbb{R}
\cup \{ \infty, -\infty \}$ where $a<0<b$, and such that \(\vec{x}(\vec{x}_0,0)=\vec{x}_0\) and
$\frac{d}{dt}\vec{x}(\vec{x}_0,t) = f(\vec{x}(\vec{x}_0,t))$ for all $t \in
(a,b)$.
The \emph{Lie derivative} of a continuously differentiable function $p:\mathbb{R}^n \to \mathbb{R}$ with respect to vector field $f$ is defined as $p' \equiv \sum_{i=1}^n \frac{\partial p}{\partial x_i} f_i$
and equals the time-derivative of $p$ evaluated along the solutions to the system $\D{\vec{x}} = f(\vec{x})$~\cite{DBLP:journals/jar/Platzer17,DBLP:journals/jacm/PlatzerT20}.

\paragraph{Semi-algebraic Sets.}
A set $S \subseteq \mathbb{R}^n$ is \emph{semi-algebraic} iff it is characterized by a finite boolean combination of polynomial equations and inequalities:
\begin{equation}
\bigvee_{i=1}^{l}\left( \bigwedge_{j=1}^{m_i}~p_{ij} < 0~\land~\bigwedge_{j=m_i+1}^{M_i}~p_{ij} = 0 \right)\enspace,
\label{eq:semialg}
\end{equation}
where $p_{ij} \in \mathbb{R}[x_1,\dots,x_n]$ (i.e. $p_{ij}$ are multivariate polynomials in the indeterminates $x_1,\dots,x_n$, with real coefficients).
By quantifier elimination, every first-order formula of real arithmetic characterizes a semi-algebraic set and can be expressed in the form~\rref{eq:semialg}, see e.g. Mishra~\cite[\S 8.6]{Mishra}.
With an abuse of notation, this article uses formulas and the sets they characterize interchangeably.

\paragraph{Continuous Invariants in Verification.}
Safety specifications for ODEs and hybrid  systems can be rigorously verified in formal logics, such as \emph{differential dynamic logic} (\dL)~\cite{DBLP:journals/jar/Platzer08,DBLP:conf/lics/Platzer12a,DBLP:journals/jar/Platzer17} as implemented in the \KeYmaeraX proof assistant~\cite{Fulton2015}
and \emph{hybrid Hoare logic}~\cite{DBLP:conf/aplas/LiuLQZZZZ10} as implemented in the \mbox{HHL prover}~\cite{DBLP:conf/icfem/WangZZ15}.
The use of appropriate continuous invariants is key to these verification approaches as they allow the complexities of the continuous dynamics to be handled rigorously even for ODEs without closed-form solutions.
For example, the \dL formula $\Init \limply [\D{\vec{x}} = f(\vec{x})~\&~Q]~\Safe$ states that the safety property $\Safe$ is satisfied throughout the continuous evolution of the system $\D{\vec{x}} = f(\vec{x})~\&~Q$ whenever the system begins its evolution from a state satisfying $\Init$.
The invariant reasoning principle for verifying such a safety property is given by the following sound rule of inference in \dL, with three premisses above the bar and the conclusion below:
\[
    (\mathrm{Safety})~\frac{\Init \limply I \quad~\quad I\limply [\D{\vec{x}} = f(\vec{x})~\&~Q]~I \quad~\quad I\limply\Safe}{\Init \limply [\D{\vec{x}} = f(\vec{x})~\&~Q]~\Safe}\enspace.
\]

In this rule, the first and third premiss respectively state that the initial set $\Init$ is contained within the set $I$, and that $I$ lies entirely inside the safe set of states $\Safe$.
The second premiss states that $I$ is a \emph{continuous invariant}, i.e. $I$ is maintained throughout the continuous evolution of the system whenever it starts inside $I$, that is, the following \dL formula is true in all states:
\begin{equation}
I\limply [\D{\vec{x}} = f(\vec{x})~\&~Q]~I\enspace.
\label{eq:invariance}
\end{equation}

Thus, the problem of verifying safety properties of ODEs reduces to finding an invariant $I$ that can be \emph{proved} to satisfy all three premisses.
Semantically, a continuous invariant can also be defined as follows.

\begin{definition}[Continuous invariant]
\label{def:continv}
Given a system $\D{\vec{x}} = f(\vec{x})~\&~Q$, the set
$I \subseteq \mathbb{R}^n$ is a continuous invariant iff the following statement holds:\footnote{To simplify notation, $\forall\,t\geq0$ is implicitly assumed to quantify over all times $t \geq 0$ in the maximal interval of existence of the ODE solution from $\vec{x}_0$, i.e., where $ \vec{x}(\vec{x}_0,t)$ is defined.}
\[
    \forall\,\vec{x}_0 \in I~\forall\,t\geq 0:~\big((\forall\, \tau\in[0,t]:~\vec{x}(\vec{x}_0,\tau) \in Q) \implies \vec{x}(\vec{x}_0,t) \in I\,\big)\enspace.
\]
\end{definition}

For any given set of initial states $\Init \subseteq \mathbb{R}^n$, a
continuous invariant $I$ such that $\Init \subseteq I$ provides a \emph{sound
over-approximation} of the states reachable by the system from $\Init$ by
following the solutions to the ODEs within the evolution domain constraint $Q$.
Indeed, the exact set of states reachable by a
continuous system from $\Init$ provides the \emph{smallest} such invariant.\footnote{Unfortunately, reachable sets rarely have a simple description as semi-algebraic sets.}
While~\rref{def:continv} above features the solution $\vec{x}(\vec{x}_0,t)$, which
may not be available explicitly, a crucial advantage afforded by continuous
invariants is the possibility of checking whether a given set is a continuous
invariant \emph{without computing the solution}, i.e. by working \emph{directly
with the ODEs}.

\section{Sound Invariant Checking and Generation}
\label{sec:checking}
The problem of \emph{checking} whether a semi-algebraic set $I\subseteq \mathbb{R}^n$ is a continuous invariant of a polynomial system of ODEs $\D{\vec{x}}=f(\vec{x})\,\&\, Q$ was shown to be \mbox{\emph{decidable}} by Liu, Zhan, and Zhao~\cite{DBLP:conf/emsoft/LiuZZ11}.
This decision procedure, henceforth referred to as LZZ, provides a way of automatically checking continuous invariants~\rref{eq:invariance} by exploiting facts about higher-order Lie derivatives of multivariate polynomials appearing in the syntactic description of $I$ and the Noetherian property of the ring $\mathbb{R}[\vec{x}]$ \cite{DBLP:journals/cl/GhorbalSP17,DBLP:conf/emsoft/LiuZZ11}; its implementation requires an algorithm for constructing Gr\"obner bases~\cite{DBLP:books/daglib/0091062}, as well as a decision procedure for the universal fragment of real arithmetic~\cite{roy1996basic}.
A logical alternative for invariant checking is provided by the complete \dL axiomatization for differential equation invariants~\cite{DBLP:journals/jacm/PlatzerT20}.
Whereas using LZZ results in a \textsf{yes/no} answer to an invariance question~\rref{eq:invariance}, \dL makes it possible to construct a \emph{formal proof of invariance} from a small set of ODE axioms~\cite{DBLP:journals/jacm/PlatzerT20} whenever the property holds (or a refutation whenever it does not).

\subsection{Invariant Generation with Template Enumeration}
Given a means to perform invariant checking with real arithmetic, an obvious solution to the invariant generation problem (which has been suggested by numerous authors \cite{DBLP:conf/emsoft/LiuZZ11,DBLP:journals/fmsd/PlatzerC09,DBLP:conf/issac/SturmT11}) involves the \emph{method of template enumeration}, which yields a theoretically complete semi-algorithm, in the sense that it terminates with a positive answer iff that is possible with the given templates. A template is a parametric formula, such as
    \[
        \textcolor{red}{a_0} + \textcolor{red}{a_1}x + \textcolor{red}{a_2}y + \textcolor{red}{a_3}x^2 + \textcolor{red}{a_4}xy + \textcolor{red}{a_5}y^2 < 0 \land \textcolor{red}{b_0} + \textcolor{red}{b_1}x + \textcolor{red}{b_2}y \geq 0\enspace,
    \]
composed from polynomials in the state variables (in this example $x,y$) with symbolic coefficients (here $\textcolor{red}{a_0}$,$\textcolor{red}{a_1}$,$\textcolor{red}{a_2}$,$\textcolor{red}{a_3}$,$\textcolor{red}{a_4}$,$\textcolor{red}{a_5}$ and $\textcolor{red}{b_0}$,$\textcolor{red}{b_1}$,$\textcolor{red}{b_2}$), which are interpreted over the reals.
All it takes \emph{in theory} is to exhaustively enumerate parametric templates matching \emph{all} real arithmetic formulas describing all semi-algebraic sets, and use a quantifier elimination algorithm (such as CAD~\cite{Collins1975}) to identify whether choices for the template parameters exist that meet the required arithmetic constraints. While templates make this British Museum Algorithm-like approach more successful than, e.g. exhaustively enumerating all proofs \cite{Herbrand30}, the method is nevertheless quite impractical for the resulting real arithmetic \cite{DBLP:conf/itp/Platzer12}. To appreciate why, let us only remark that quantifier elimination algorithms for real arithmetic used in practice have doubly-exponential time complexity in the number of variables~\cite{DBLP:conf/dimacs/Renegar90}. Template enumeration treats every monomial coefficient in the template as a fresh variable, leading to exponentially many real arithmetic variables, which makes this approach highly unscalable.
In practice, invariant generation is achieved by using incomplete---but considerably more efficient---generation methods. These methods are numerous and vary considerably in their strengths and limitations, creating a wide spectrum of possible trade-offs in performance, the quality, and the form of invariants that one can generate. Effectively navigating this spectrum is an important practical challenge that this article seeks to address.

\subsection{Soundness: Proof Assistants and Invariant Generation}
There are a number of design decisions that can be exercised in how reasoning
with continuous invariants is performed within a deductive verification
framework.
A fundamental design decision is how tightly \begin{enumerate*}[label=(\roman*)] \item continuous invariance checking and \item continuous invariant generation \end{enumerate*} are to be coupled with the implementation of the prover.
This space of design choices is exemplified by the HHL prover and the \KeYmaeraX prover.

The HHL prover~\cite{DBLP:books/sp/17/ChenHTWYZZZ17,DBLP:conf/icfem/WangZZ15} implements
\begin{enumerate*}[label=(\roman*)]
\item the LZZ decision procedure for invariant checking and
\item the method of template enumeration for invariant generation based on real quantifier elimination and Gr\"obner bases.
\end{enumerate*}
From the perspective of the HHL prover, these are \emph{trusted external oracles} for checking the validity of statements about continuous invariance; trusting the output of the HHL prover includes trusting the implementation of its LZZ procedure and the invariant generator (and any arithmetic tool either of them use).

\begin{figure}[h!]
	\begin{subfigure}[t]{0.5\textwidth}
		\centering
\begin{tikzpicture}[scale=0.85, every node/.style={scale=0.85}]
\draw[rounded corners, fill=lsred!30]  (-4.5,2.7352) rectangle (-0.5,1.7056) node (v2) {};
	\node at (-2.5,1.9704) {{\bf \textsf{LZZ procedure}}};
\node[opacity=0] at (-2.5,2.2204) {assistant};
\draw[rounded corners, fill=lsred!30]  (-4.5,5) node (v1) {} rectangle (-0.5,3.9116);
	\node at (-2.5,4.1764) {{\bf \textsf{prover core}}};

	\node at (-2.5,2.4704) {\it \textsf{soundness-critical}};
	\node at (-2.5,4.6764) {\it \textsf{soundness-critical}};
\node at (-1.1764,3.294) {{\textsf{yes/no}}};
\node at (-3.8704,3.294) {{\textsf{invariant?}}};
\draw[->, >=stealth, thick] (-3.0292,3.8764) --  (-3.0292,2.7352);
\draw[->, >=stealth, thick]  (-2.0292,2.7352) -- (-2.0292,3.8764) ;
\draw[->, >=stealth, thick] (-2.5,5.5) -- (-2.5,5);

\node at (-2.5,5.8236) {$I \to [x' = f(x) \, \& \,Q]\,I$};
\node[rotate=90] at (-5,4.47) {\bf HHL porver};
\node[rotate=0] at (-2.5,1.3824) {\it External oracle};
\end{tikzpicture}

		\caption{PVS-style \label{fig:pvs}}
	\end{subfigure}
	\begin{subfigure}[t]{0.5\textwidth}
		\centering

\begin{tikzpicture}[scale=0.85, every node/.style={scale=0.85}]

\draw[rounded corners, fill=lsred!30]  (-4.5,5) rectangle (-0.5,2.1172);
\node at (-2.5,2.4996) {{\bf \dL core}};
\node[opacity=0] at (-2.5,2.7496) {assistant};
\draw[dashed, rounded corners, fill=lsrealgreen!60]  (-4.5,5) rectangle (-0.5,3.7056);
\node at (-2.5,4.1764) {{\bf \dL tactics}};

	\node at (-2.5,2.9996) {\it \textsf{soundness-critical}};
	\node at (-2.5,4.6764) {\it \textsf{non-critical}};
\node at (-2.5,1.2936) {{\it Proof/refutation from {\bf \dL} axioms}};
\draw[->, >=stealth, thick] (-2.5,3.8764) -- (-2.5,3.3232);
\draw[->, >=stealth, thick] (-2.5,5.5) -- (-2.5,5);
\draw[->, >=stealth, thick] (-2.5,2.1172) -- (-2.5,1.6172);
\node at (-2.5,5.8236) {$I \to [x' = f(x) \, \& \,Q]\,I$};
\node[rotate=90] at (-5,3.6468) {\bf \KeYmaeraX};
\end{tikzpicture}
		\caption{LCF-style  \label{fig:lcf}}
	\end{subfigure}
	\

		\caption{Alternative prover architectures for \emph{checking} conjectured continuous invariants, i.e. formulas for the form $I\limply [\D{\vec{x}} = f(\vec{x})~\&~Q]~I $~\label{fig:arch}}
\end{figure}
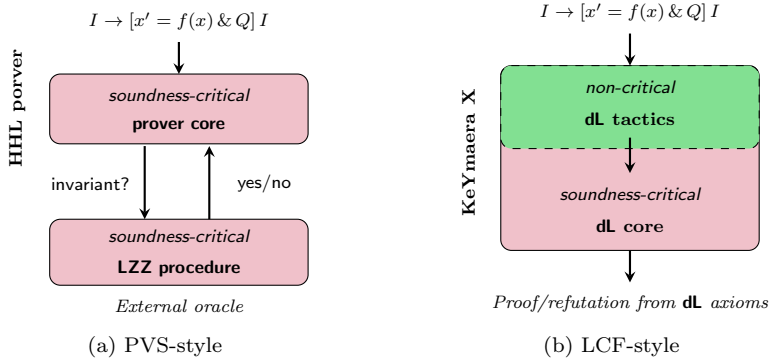

In contrast, \KeYmaeraX~\cite{Fulton2015} pursues an LCF-style approach, seeking to minimize the soundness-critical code that needs to be trusted in its output~\cite{MitschP20}. For continuous invariants, it achieves this by
\begin{enumerate*}[label=(\roman*)]
    \item checking invariance within the axiomatic framework of \dL (rather than trusting external checking procedures) and
\item accepting \emph{conjectured invariants} generated from a variety of sources but \emph{separately checking} the result.
\end{enumerate*}
Invariant checking in \KeYmaeraX is automatic \cite{DBLP:journals/jacm/PlatzerT20}, which is made possible by the use of specialized
proof \emph{tactics}~\cite{Fulton2017};
these additionally allow it to use a variety of other
(incomplete, but computationally inexpensive) methods for proving continuous invariance~\cite{DBLP:journals/cl/GhorbalSP17}.

\begin{remark}
	The difference between these two approaches (\rref{fig:arch}) is broadly analogous to the use of trusted decision procedures in PVS~\cite{DBLP:conf/fm/DenmanM14} and oracles in HOL~\cite{DBLP:conf/itp/BohmeW10,weber2011smt} on the one hand, and LCF-style proof reconstruction (e.g. in Isabelle~\cite{DBLP:journals/entcs/Weber06}) on the other.
\end{remark}
\begin{remark}
	\KeYmaeraX also supports witness checking for the universal fragment of real arithmetic \cite{DBLP:conf/cade/PlatzerQR09} resulting from ODE invariance checking \cite{DBLP:journals/jacm/PlatzerT20}.
	In theory, this leads to a complete LCF-style approach, but in practice, the performance of real arithmetic witness generation is only competitive with second-tier quantifier elimination \cite{DBLP:conf/cade/PlatzerQR09}.
\end{remark}

\subsection{Syntactic Representation of Invariants}

A subtle issue that arises when interfacing with provers like \KeYmaeraX or the HHL prover is which terms can be \emph{syntactically} represented in the prover.
The choice of representation limits the kinds of invariants that can be described (or generated), but it is an important consideration for computational efficiency and soundness purposes.
For example, Noetherian functions support a sound and complete axiomatization of invariants in \dL~\cite{DBLP:journals/jacm/PlatzerT20} but can lead to undecidable arithmetic.
Rational functions and roots could be supported~\cite{DBLP:conf/cade/BohrerFP19} but would increase the complexity of the required symbolic computations.
For decidability of the invariance and arithmetic questions, this article only considers semi-algebraic invariants, i.e., those built from polynomials as in~\rref{eq:semialg}.

A similar issue arises even when restricted to polynomial terms.
Na\"ively, for maximum flexibility, one would like to describe invariants using polynomials $p\in\polynomials{\reals}{x}$ that have arbitrary real-valued coefficients.
In practice though, only \emph{computable} subfields $K$ of $\reals$ can be effectively represented and used on a computer.
Thus, any computational tool must necessarily work with polynomials $p\in\polynomials{K}{x}$ over some choice of representation for the field of coefficients $K$.
Real algebraic numbers $K=\bar{\rationals}$ would work as coefficients, but they increase the complexity of symbolic computations due to the added need to work with polynomial ideal arithmetic for coefficients and can also lead to some subtleties with the non-differentiability of the resulting root function itself~\cite{DBLP:conf/cade/BohrerFP19}.
On the other extreme, floating point numbers are computationally efficient but they do not form a field, and would also cause numerical errors that make it harder to obtain sound and exact answers in the end.
For these reasons, \KeYmaeraX works with polynomials $p\in\polynomials{\rationals}{x}$ that have rational coefficients.\footnote{In practice, some generation methods may need to internally use floating point arithmetic when interfacing with numerical solvers, but Pegasus then applies rounding procedures to obtain polynomials with rational coefficients.}
This results in fast evaluations and symbolic computations, and a reasonable (although nontrivial) complexity for the resulting real arithmetic validity decision problem.
Many invariant generation techniques described in this article are fairly general and agnostic to the precise choice of field $K$.
Thus, the rest of this article elides this subtlety and describes the invariant generation algorithms over $p \in\polynomials{\reals}{x}$, i.e., with $\reals$ as the coefficient field.

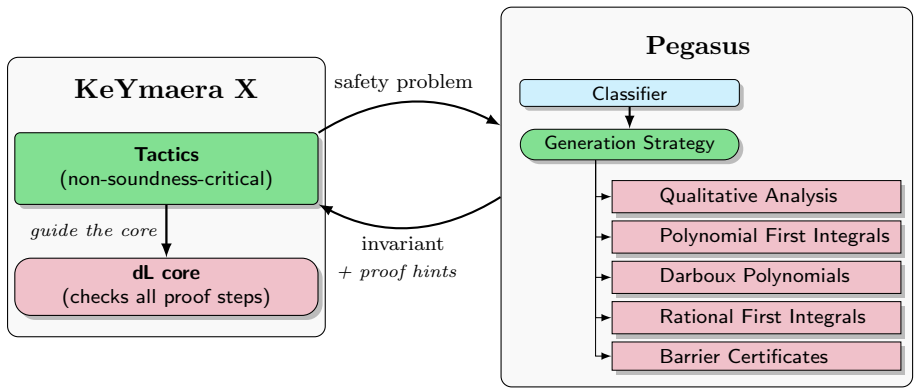
\begin{figure}[bht]
\begin{tikzpicture}[
  scale=0.95, every node/.style={scale=0.95},
  level 1/.style={sibling distance=20mm},
  edge from parent/.style={->,draw},
  >=latex]
\node[draw,fill=gray!5,rectangle,rounded corners,minimum width=4.4cm,minimum height=3.9cm] (keymaerax) {};
\node[below=.2cm of keymaerax.north] {\large\textbf{\KeYmaeraX}};
\node[draw,fill=lsrealgreen!60,rectangle,drop shadow,rounded corners=2pt,minimum width=3.5cm,minimum height=1cm,text width=4cm,font=\sffamily,align=center,below=1cm of keymaerax.north] (tactics) {\textbf{Tactics}\\ (non-soundness-critical)};
\node[draw,fill=lsred!30,rectangle,drop shadow, rounded corners=6pt,minimum width=3.5cm,minimum height=.5cm,text width=4cm,align=center,font=\sffamily,below=.7cm of tactics] (core) {\textbf{\dL core}\\ (checks all proof steps)};
    \draw[->, line width=0.3mm] (tactics) -- node[left] {\scriptsize \textit{guide the core}} (core);

\node[draw,fill=gray!5,rectangle,rounded corners,minimum width=5.7cm,minimum height=5.3cm,right=2.3cm of keymaerax] (pegasus) {};
\node[below=.2cm of pegasus.north] {
\centering

    \vspace{5cm}
\begin{tikzpicture}[
  level distance=0.5cm,
  level 1/.style={sibling distance=20mm},
  edge from parent/.style={->,thick,draw},
  >=latex]

\node at (0,0.7) {\large\textbf{Pegasus}};
\node[root, xshift=-1cm] {Classifier}
  child {node[level 2] (c2) {Generation Strategy}};

\begin{scope}[every node/.style={level 3}]
	\node [below=0.1cm of c2, xshift=50pt, yshift=-5pt, minimum width=4cm] (c21) {\mbox{Qualitative Analysis}};
	\node [below=0.1cm of c21, minimum width=4cm] (c22) {\mbox{Polynomial First Integrals}};
	\node [below=0.1cm of c22,minimum width=4cm] (c23) {\mbox{Darboux Polynomials}};
	\node [below=0.1cm of c23,minimum width=4cm] (c24) {\mbox{Rational First Integrals}};
	\node [below=0.1cm of c24,minimum width=4cm] (c25) {\mbox{Barrier Certificates}};
\end{scope}

\foreach \value in {1,...,5}
  \draw[->] (c2.205) |- (c2\value.west);

\end{tikzpicture}
};
     \node[text width=3cm] at (3.83,1.59) {\footnotesize safety problem};
     \node[text width=3cm] at (3.85,-1.05) {\scriptsize \textit{+ proof hints}};
     \node[text width=3cm] at (4.2,-0.65) {\small invariant};
    \draw [thick, -latex] ($(tactics.east)+(0,0.5)$) to [bend left] ($(pegasus.west)+(0,1)$);
    \draw [thick, -latex] ($(pegasus.west)-(0,0)$) to [bend left] ($(tactics.east)-(0,0.5)$);
\end{tikzpicture}
\caption{Sound invariant generation: invariant generator analyzes safety problem to provide invariants and proof hints to tactics; the invariants are formally verified to be correct within the soundness-critical \dL core}
\label{fig:architecture}
\end{figure}

\section{Invariant Generation Methods in Pegasus}
\label{sec:pegasus}
Pegasus is a continuous invariant generator implemented in the Wolfram Language with an interface accessible through both Mathematica and \KeYmaeraX.\footnote{Pegasus (\url{http://pegasus.keymaeraX.org/})
is linked to \KeYmaeraX through the Mathematica interface of \KeYmaeraX, which translates between the internal data structures of the prover core and the Mathematica data structures.}
When \KeYmaeraX is faced with a continuous safety verification problem that it is unable to prove directly, it automatically invokes Pegasus to help find an appropriate invariant (if possible).
\KeYmaeraX checks \emph{all} the invariants it is supplied with---\emph{including those provided by Pegasus} (see~\rref{fig:architecture}).
This design ensures that any correctness issues in Pegasus cannot compromise the soundness of \KeYmaeraX.
It also presents implementation opportunities:%

\begin{enumerate}
\item Pegasus can freely integrate numerical procedures and heuristic methods while providing \emph{best-effort} guarantees of correctness. Final correctness checks for the generated invariants are left to the purview of \KeYmaeraX.\footnote{Naturally, the output from Pegasus can also be checked using a trusted implementation of the LZZ decision procedure before anything is returned. When used with \KeYmaeraX, though, this additional (soundness-critical) check is unnecessary.}

\item Pegasus records \emph{proof hints} corresponding to the various methods that were used to generate continuous invariants.
These hints enable \KeYmaeraX to build more efficient shortcut proofs of continuous invariance~\cite{DBLP:journals/cl/GhorbalSP17}.
\end{enumerate}

Pegasus currently implements an array of powerful invariant generation methods,
which we describe below, beginning with a large family of related methods that are based
on \emph{qualitative analysis}, which can be best explained using the machinery of
\emph{discrete abstraction} of continuous systems. We first briefly recall the main idea
behind this approach.

\subsection{Exact Discrete Abstraction}\label{sec:ExactDiscreteAbstraction}
Discrete abstraction is the subject of numerous works~\cite{Alur2000,Tiwari2008FMSD,TiwariKhanna2002HSCC}.
Briefly, the steps are:
\begin{enumerate*}[label=(\roman*)] \item discretize the continuous state space of a system by defining \emph{predicates} that correspond to discrete states, \item \label{discabsstep2}compute a (local) transition relation between the discrete states obtained from the previous step, yielding a discrete transition system which abstracts the behavior of the original continuous system, and finally \item compute reachable sets in the discrete abstraction to obtain an over-approximation of the reachable sets of the continuous system.
\end{enumerate*}

A discrete abstraction is \emph{sound} iff the relation computed in step~\ref{discabsstep2} has a transition between two discrete states whenever there is a corresponding trajectory of the original continuous system between the two neighboring sets corresponding to those discrete states.
The abstraction is \emph{exact} iff these are the \emph{only} transitions computed in step~\ref{discabsstep2}.
Soundness of the discrete abstraction guarantees that any invariant extracted from the discretization corresponds to an invariant for the original system.
Exactness implies that no invariants are lost that are representable in the abstraction at all.
\begin{figure}[ht]
    {\setlength{\belowcaptionskip}{0pt}
      \centering
      ~ %
      \begin{subfigure}[b]{0.45\textwidth}
         \centering
\begin{tikzonimage}[scale=0.3]{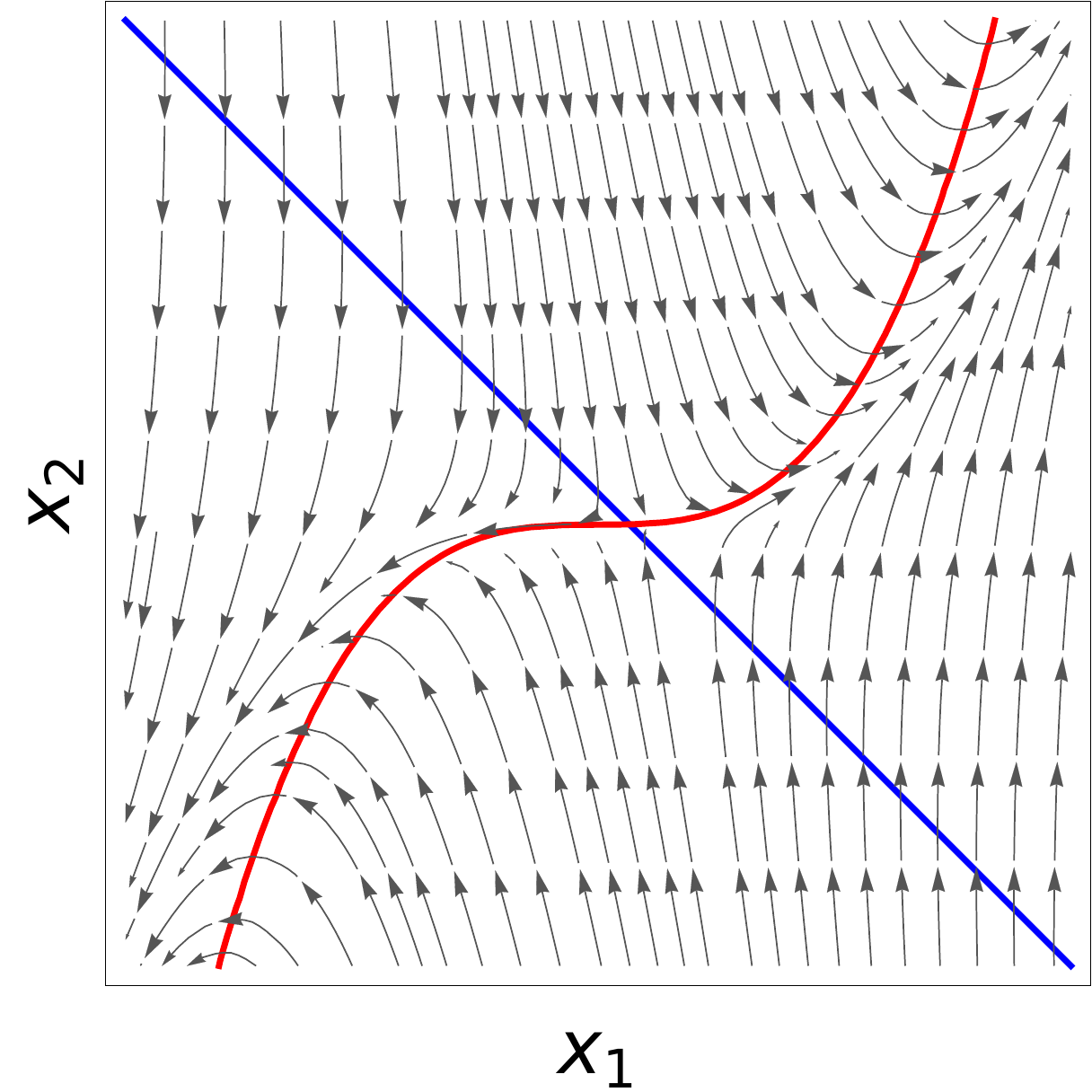}
    \node[draw,fill=white,circle,inner sep=0pt,font=\sffamily\normalsize\bfseries, scale=0.65] at (0.57,0.52) {$S_5$};
    \node[draw,fill=white,circle,inner sep=0pt,font=\sffamily\normalsize\bfseries, scale=0.65] at (0.24,0.85) {$S_1$};
    \node[draw,fill=white,circle,inner sep=0pt,font=\sffamily\normalsize\bfseries, scale=0.65] at (0.85,0.8) {$S_3$};
    \node[draw,fill=white,circle,inner sep=0pt,font=\sffamily\normalsize\bfseries, scale=0.65] at (0.55,0.9) {$S_2$};
    \node[draw,fill=white,circle,inner sep=0pt,font=\sffamily\normalsize\bfseries, scale=0.65] at (0.3,0.6) {$S_4$};
    \node[draw,fill=white,circle,inner sep=0pt,font=\sffamily\normalsize\bfseries, scale=0.65] at (0.85,0.4) {$S_6$};
    \node[draw,fill=white,circle,inner sep=0pt,font=\sffamily\normalsize\bfseries, scale=0.65] at (0.23,0.2) {$S_7$};
    \node[draw,fill=white,circle,inner sep=0pt,font=\sffamily\normalsize\bfseries, scale=0.65] at (0.55,0.3) {$S_8$};
    \node[draw,fill=white,circle,inner sep=0pt,font=\sffamily\normalsize\bfseries, scale=0.65] at (0.89,0.2) {$S_9$};
    \node[draw=red,rotate=61,fill=red!10, scale=0.7, inner sep=1pt] at (0.3,0.35) {\color{red} $p_1=0$};
    \node[draw=blue,rotate=-45,fill=blue!10, scale=0.7, inner sep=1pt] at (0.4,0.7) {\color{blue} $p_2=0$};
\end{tikzonimage}
          \caption{Discretization with $p_1$,~$p_2$ $\sim 0$}
         \label{fig:discretization}
      \end{subfigure}
       \qquad
      \begin{subfigure}[b]{0.45\textwidth}
         \centering
\begin{tikzpicture}[transform shape, scale=0.65,node distance=2cm,state node/.style={circle,draw,font=\sffamily\normalsize\bfseries}]
   \node[state node] (M1) [inner sep=0pt] {$S_1$};
   \node[state node] (M2) [right of=M1, inner sep=0pt] {$S_2$};
   \node[state node] (M3) [right of=M2, inner sep=0pt] {$S_3$};
   \node[state node] (M4) [below of=M1, inner sep=0pt] {$S_4$};
   \node[state node] (M5) [right of=M4, inner sep=0pt] {$S_5$};
   \node[state node] (M6) [right of=M5, inner sep=0pt] {$S_6$};
   \node[state node] (M7) [below of=M4, inner sep=0pt] {$S_7$};
   \node[state node] (M8) [right of=M7, inner sep=0pt] {$S_8$};
   \node[state node] (M9) [right of=M8, inner sep=0pt] {$S_9$};
  \draw[->,>=stealth,bend right=20] (M1) to node [above=1pt] { } (M4);
  \draw[->,>=stealth,bend right=20] (M7) to node [above=1pt] { } (M4);
  \draw[->,>=stealth,bend right=20] (M2) to node [below=2pt] { } (M1);
  \draw[->,>=stealth,bend right=25] (M8) to node [sloped, anchor=south, above=1pt] { } (M7);
  \draw[->,>=stealth,bend right=25] (M8) to node [sloped, anchor=south, above=1pt] { } (M9);
  \draw[->,>=stealth,bend right=25] (M9) to node [sloped, anchor=south, above=1pt] { } (M6);
  \draw[->,>=stealth,bend right=25] (M2) to node [sloped, anchor=south, above=1pt] { } (M3);
  \draw[->,>=stealth,bend right=25] (M3) to node [sloped, anchor=south, above=6pt] { } (M6);
\end{tikzpicture}
          \vspace{15pt}
         \caption{Sound discrete abstraction}
         \label{fig:abstraction}
      \end{subfigure}}%
     \caption{Discrete abstraction of a two-dimensional system}
      \label{fig:abstr}
   \end{figure}
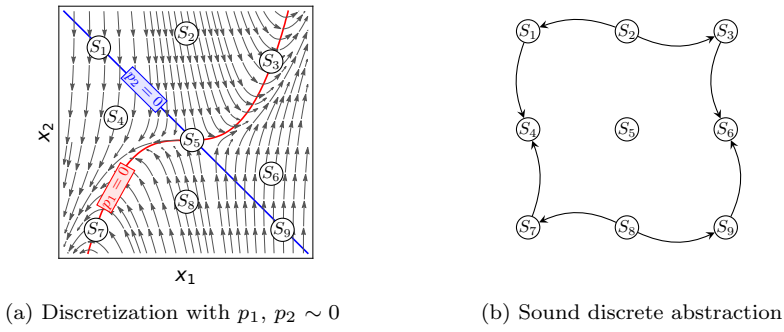

Figure~\ref{fig:abstr} illustrates a discretization of a system of ODEs (\rref{fig:discretization}), which results in $9$ discrete states in a sound and exact abstraction (\rref{fig:abstraction}).
The state space is discretized using predicates built from sign conditions on polynomials, $p_1,p_2 \in
\mathbb{R}[x_1,x_2]$.\footnote{\emph{Sign conditions} on a polynomial $p$ are atomic formulas $p<0$, $p=0$, and $p>0$.}
The discrete states of the abstraction are given by formulas such as $S_1 \equiv
p_1<0 \land p_2=0$, $S_2 \equiv p_1<0 \land p_2>0$, and so on.
The question whether there should be a discrete transition from $S_1$ to $S_2$ in the abstraction may be equivalently cast as
the following question: is $S_1$ a continuous invariant of the system
$\vec{x}' = f(\vec{x})$ under evolution domain constraint $S_1 \lor S_2$, i.e. is the following \dL formula valid?
\[ 
S_1 \to [\vec{x}' = f(\vec{x})\,\&\,S_1 \lor S_2]\, S_1\enspace.
\]

This question can be answered with a decision procedure such as LZZ or formally proved/disproved using \dL, as discussed in Section~\ref{sec:checking}. If $S_1$ is a continuous invariant under this evolution domain constraint, then there are no states satisfying $S_1$ from which the system continuously evolves into a state satisfying $S_2$ along a trajectory that remains within the union $S_1 \cup S_2$ and thus there should not be a transition from $S_1$ to $S_2$ if the discrete abstraction is to be exact; on the other hand, if $S_1$ is not a continuous invariant, then there must be such a transition if the abstraction is to be sound.

The ability to construct sound and exact discrete abstractions~\cite{DBLP:conf/vmcai/SogokonGJP16} has an important consequence:
if an appropriate semi-algebraic continuous invariant $I$ exists at all, it can always be extracted from a discrete abstraction built from discretizing the state space using sign conditions on the polynomials describing~$I$.
The problem of (semi-algebraic) invariant generation therefore reduces to finding appropriate polynomials whose sign conditions can yield suitable discrete abstractions and computing reachable states in these abstractions.

\begin{remark}
    Reachable sets (from the initial states) in discrete abstractions are the
    smallest invariants with respect to $\subseteq$ (set inclusion) that are representable in that abstraction. The smallest invariant is the most informative because it
    allows one to prove the most safety properties, but it may not be the most
    useful invariant in practice.
    In particular, one often wants to work with invariants that have \emph{low
    descriptive complexity} and are easy to prove in the formal proof calculus.
    This leads naturally to consider alternative ways of extracting invariants.
    Pegasus is able to extract reachable sets of discrete abstractions, but
    favours less costly techniques, such as \emph{differential saturation}~\cite{DBLP:journals/fmsd/PlatzerC09},
    which often succeed in more quickly extracting more conservative invariants.
\end{remark}

Finding ``good'' polynomials that can abstract the system in useful ways
and allow proving properties of interest is generally difficult.
While abstraction using predicates that are extracted from the
verification problem itself can be surprisingly effective, in certain
cases useful predicates may not be syntactically extracted from the problem statement.
In order to improve the quality of discrete abstractions, Pegasus employs a
separate \emph{classifier}, which extracts
features from the verification problem which can then be used to suggest
polynomials that are more tailored to the problem at hand.
Certain systems have structure that, to a human expert, might suggest an ``obvious''
choice of good predicates. Below we sketch some basic examples of what is currently possible.

\subsection{Targeted Qualitative Analysis}

As a motivating example, consider the class of one-dimensional
ODEs \mbox{$x'=f(x)$}, where $f\in \mathbb{R}[x]$. A standard way of studying qualitative
behavior in these systems is to inspect the graph of the function $f(x)$~\cite{Strogatz2001}.
Figure~\ref{fig:onedim} illustrates such a graph of $f(x)$, along with a vector
field induced by such a system on the real line.
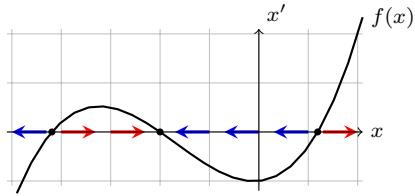
\begin{figure}[h!]
    \centering
    \begin{tikzpicture}[domain=-4.9:2.1, scale=0.65]
        \draw[very thin,color=lightgray] (-5.1,-1.1) grid (2.1,2.1);
        \draw[>=stealth, very thick, color=red, ->] (-4.0,0) -- (-3.3,0);
        \draw[>=stealth, very thick, color=red, ->] (-3.0,0) -- (-2.3,0);
        \draw[>=stealth, very thick, color=red, ->] (1.3,0) -- (2.0,0);
        \draw[>=stealth, very thick, color=blue, <-] (0.3,0) -- (1.0,0);
        \draw[>=stealth, very thick, color=blue, <-] (-0.7,0) -- (0,0);
        \draw[>=stealth, very thick, color=blue, <-] (-1.7,0) -- (-1,0);
        \draw[>=stealth, very thick, color=blue, <-] (-5,0) -- (-4.3,0);
        \draw[->] (-5.1,0) -- (2.1,0) node[right] {$x$};
        \draw[->] (0,-1.2) -- (0,2.1) node[above right] {$\D{x}$};
        \draw[color=black, thick] plot[id=sin,mark=none] (\x, {0.1*(\x)^3+0.5*(\x)^2+0.1*(\x)-1})  node[right] {$f(x)$};
        \draw[black, fill=black] (1.19,0) circle (.5ex);
        \draw[black, fill=black] (-4.19,0) circle (.5ex);
        \draw[black, fill=black] (-2,0) circle (.5ex);
    \end{tikzpicture}
    \caption{Qualitative analysis of one-dimensional ODEs $\D{x} = f(x)$}
    \label{fig:onedim}
\end{figure}
The ODE $\D{x}=f(x)$ is at an \emph{equilibrium} without any motion at points where $f(x)=0$.
By computing the real roots of the polynomial in the right-hand side, i.e the real
roots $r_1,\dots,r_k \in \mathbb{R}$ of $f(x)$, we may form a list of polynomials
$x-r_1,\dots,x-r_k$ that can be used for an \emph{algebraic decomposition} of $\mathbb{R}$
into invariant subregions corresponding to real intervals from which an over-approximation
of the reachable set can be constructed. Such an algebraic decomposition can be further
refined by augmenting the list of polynomials with $x-b_1,\dots,x-b_l$,
where $b_1,\dots,b_l \in \mathbb{R}$ are the boundary points
of the initial set in the safety specification. From this augmented list, one can
exactly construct the \emph{reachable set} of the system by computing the reachable set of the corresponding exact abstraction.

\begin{remark}
    If $\D{x} = f(x)$ is one-dimensional, one can exploit another useful fact:
every one-dimensional system is a \emph{gradient system}, i.e. its motion is generated by a \emph{potential
function} $F(x)$
which can be computed directly by integrating $-f(x)$ with respect to $x$, i.e. \(F(x) = -\int f(x)\,dx\).
For any $k \in \mathbb{R}$, $F(x)\leq k$ defines a continuous invariant of the one-dimensional system \(\D{x} = f(x)\).
\end{remark}

In higher dimensions, the behavior of \emph{linear} systems $\D{\vec{x}} = A\vec{x}$
with a constant coefficient matrix $A$ can be studied qualitatively by examining the eigenvalues and eigenvectors\footnote{%
A vector $\vec{v}\in\reals^n$ is an \emph{eigenvector} for \emph{eigenvalue} $\lambda\in\complex$ of matrix $A\in\reals^{n\times n}$ iff \(A\vec{v}=\lambda\vec{v}\).
In direction $\vec{v}$, the ODE \(\D{\vec{x}}=A\vec{x}\), thus, converges to 0 if $\lambda<0$ or diverges if $\lambda>0$.
}  of the matrix $A$~\cite{arrowsmith1992}.
Pegasus implements methods targeted at linear systems that take advantage of facts such as these to suggest useful abstractions from which invariants can be extracted. The current strategy is similar in spirit to the abstraction methods proposed in the work of Tiwari~\cite{Tiwari2003}, and works by computing linear forms describing the invariant half-spaces in the state space of linear systems. Briefly, whenever the system matrix $A$ has a real eigenvalue $\lambda \in \mathbb{R}$, by considering an eigenvector $\vec{v}$ of the \emph{transpose} matrix $A^T$, which is associated with the eigenvalue $\lambda$ (recall that the eigenvalues of square matrices $A$ and $A^T$ are the same), one may construct the linear form $p = \vec{v}^T \vec{x}$, which has the property that~\cite[\S 2]{Tiwari2003}:
\[
	p' = \vec{v}^T \vec{x}' = \vec{v}^TA\vec{x} = (A\vec{v})^T\vec{x} = (\lambda \vec{v})^T\vec{x} = \lambda p\enspace.
\]
Such linear forms correspond to a special case of so-called \emph{Darboux polynomials}, which will be described in more detail in \rref{subsec:darboux} and have the property that $p>0$, $p=0$, and $p<0$ define invariant regions in state space (the fact that $\lambda$ is a real number also allows us to construct invariants $p \leq k$, where $k$ is an appropriately chosen offset depending on the sign of $\lambda$).

Additionally, when all the eigenvalues of the system matrix $A$ have strictly negative real parts, the origin $\vec{0}$ is asymptotically stable and one may construct a \emph{Lyapunov function} (see~\cite[Ch. 3]{SlotineLi1991},\cite[Ch. 3]{Khalil}) for the linear system by solving the \emph{Lyapunov equation} $A^TP + PA = Q$ where $Q$ is some given negative-definite matrix\footnote{An $n \times n$ matrix $Q$ is \emph{negative-definite} if it is symmetric, i.e. $Q = Q^T$, and $\vec{x}^TQ\vec{x}<0$ for all $\vec{x} \in \mathbb{R}^n \setminus \{\vec{0}\}$; a symmetric matrix $P$ is \emph{positive-definite} if $\vec{x}^TP\vec{x}>0$ for all $\vec{x} \in \mathbb{R}^n \setminus \{\vec{0}\}$.}, and the solution $P$ is positive-definite (see~\cite[Ch. 3, \S 3.5]{SlotineLi1991}); the quadratic Lyapunov function $V$ for the stable system is given by $V(\vec{x}) = \vec{x}^TP\vec{x}$. Every sub-level set $V \leq k$ defines a continuous invariant of the system; \rref{fig:ex1} (right) illustrates the kind of invariants that can be obtained by using Lyapunov functions together with invariant half-planes to perform abstraction of linear systems.

\begin{example}
\begingroup
\setlength{\intextsep}{0pt}
\setlength{\belowcaptionskip}{10pt}

The linear systems in \rref{fig:ex1} exhibit different qualitative behaviors.
The invariants (shown in blue), demonstrate unreachability of the unsafe states (shown in red) from the initial states (shown as green disks in \rref{fig:ex1}).
\begin{figure}[bhtp]
    \centering
    \begin{align*}
        &  \quad &  x_1'&=-4x_2,  &\quad \quad  x_1'&= 2x_1 - x_2, &~ x_1'&=-2x_1 + x_2,  & &  \\
        &  \quad &  x_2'&=x_1,    &\quad \quad  x_2'&= -3x_1 +x_2, &~ x_2'&=x_1 - 3x_2.   & &
    \end{align*}
    \includegraphics[width=0.8\columnwidth,trim=5 11 5 11,clip]{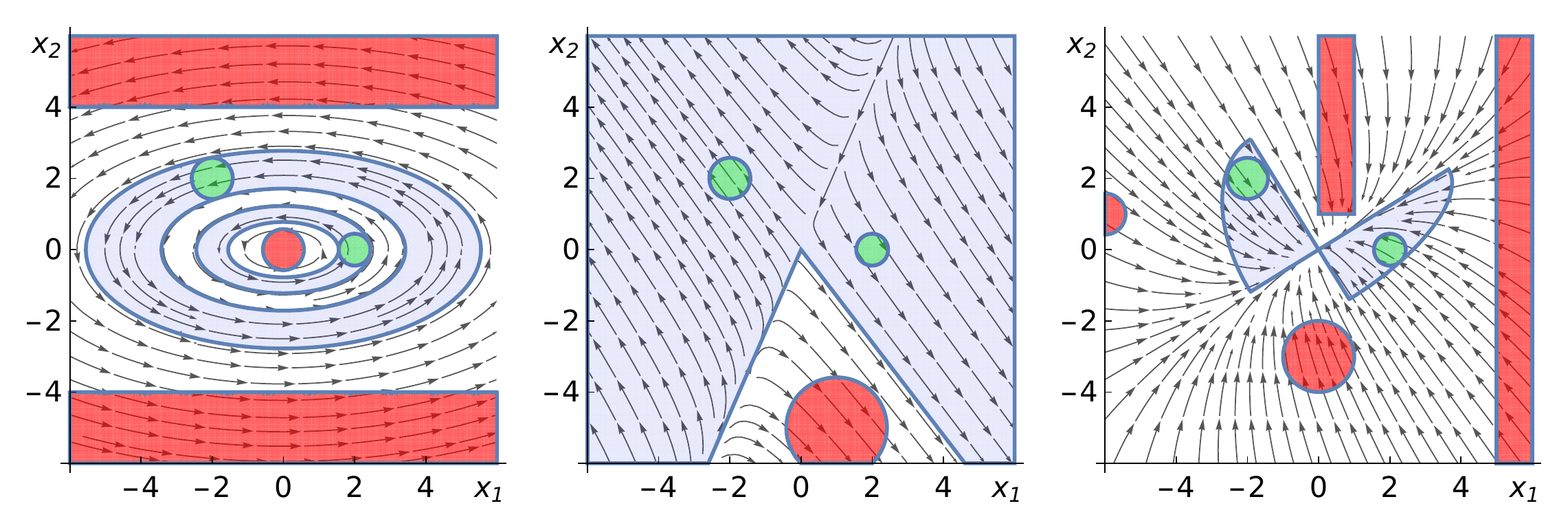}
    \caption{Automatically generated invariants for linear systems}
    \label{fig:ex1}
\end{figure}

In the leftmost system, all eigenvalues of the system matrix $A$ are purely imaginary. Pegasus generates annular invariants containing the green disks because trajectories of such systems are always elliptical.
For the middle system, the (asymptotic) behavior of its trajectories is determined by the eigenvectors of its system matrix (eigenvalues are real and of opposite sign~\cite{arrowsmith1992}). Pegasus uses these eigenvectors to generate two invariant half-planes, one for each green disc.
Invariant half-planes are also generated for the rightmost system which is asymptotically \emph{stable} (all real parts of eigenvalues are negative~\cite{arrowsmith1992}). Pegasus further refines these half-planes with suitable elliptical regions containing the green disks because elliptical regions are invariants for such systems.

\endgroup
\end{example}

\begin{dangerous}
	In textbook examples of linear systems, one usually finds matrices with eigenvalues and eigenvectors that can be described using rational numbers. However, the situation is not always that nice in practice: eigenvectors of matrices will often feature \emph{irrational} components, which in the case of the example above leads to invariant half-planes described by linear polynomials with irrational coefficients. It is therefore important to have the means of working with irrational real numbers in the invariant generator and the prover.
\end{dangerous}

In special cases when the verification problem features a purely \emph{algebraic initial set}, the strongest algebraic invariants for linear systems (i.e. the smallest continuous invariants that can be described by polynomial equalities $p=0$) can be computed following the method of Rodr\'iguez-Carbonell \& Tiwari~\cite{DBLP:conf/hybrid/Rodriguez-CarbonellT05}, which we implement in Pegasus.

\begin{remark}
Bogomolov \emph{et al.}~\cite{DBLP:conf/formats/BogomolovGHK17} introduced a technique called \emph{conic abstractions} that combines discrete abstraction of affine systems with an associated reachability analysis method.
	It is particularly powerful for systems $\vec{x}' = A\vec{x}$ in which the matrix $A$ is diagonalizable\footnote{The matrix $A$ is \emph{diagonalizable} iff it can be written as $A=PDP^{-1}$ for some invertible matrix $P$ and diagonal matrix $D$.}, where the authors' experiments suggest it outperforms other tools for linear reachability analysis, like SpaceEx~\cite{DBLP:conf/cav/FrehseGDCRLRGDM11}.
The eponymous idea behind the method is to partition state space into a number of regions (i.e., \emph{cones}), so that within each cone the change in angle of the vector field (i.e., the \emph{twisting}) is bounded by a tunable parameter $\theta$.
Given any point in the vector field, then, this construction gives a known range of possible slopes for the vector at that point.
	This is useful information for the subsequent reachability analysis---instead of simply computing the transition relation between neighboring cones, as in~\rref{sec:ExactDiscreteAbstraction}, the algorithm~\cite{DBLP:conf/formats/BogomolovGHK17} uses the twisting information to determine what portions of each cone is potentially reachable from an initial set.
We experimented with the conic abstraction method in a limited setting: bounded linear 2-dimensional systems.
The major obstacle inhibiting a complete implementation is that Mathematica's native support for polyhedra computations does not quite meet the demands of the algorithm.
Our limited implementation is not able to return an exact invariant region---instead, we produce promising visualizations of the invariant generated for two examples from~\rref{fig:ex1} (see~\rref{fig:cabs}).\footnote{The conic abstractions approach does not work directly with the leftmost example from~\rref{fig:ex1} because the example's system matrix has purely imaginary eigenvalues and is consequently not diagonalizable (a key requirement for termination of the approach~\cite{DBLP:conf/formats/BogomolovGHK17}).}
With better support for polyhedra computations, this could be an exciting direction for future implementation by interfacing Pegasus with the Parma Polyhedra Library.
\end{remark}

\begin{figure}[h!]

\centering

     \begin{subfigure}[b]{0.2\textwidth}
     \begin{align*}
     &\,\quad x_1'= 2x_1 - x_2, \\
     &\,\quad x_2'= -3x_1 +x_2. \\
     &\text{(\rref{fig:ex1} middle example)}\\~\\~\\~\\
     \end{align*}
    \end{subfigure}
    \quad
     \begin{subfigure}[b]{0.28\textwidth}
        \flushleft
        \includegraphics[width=\textwidth]{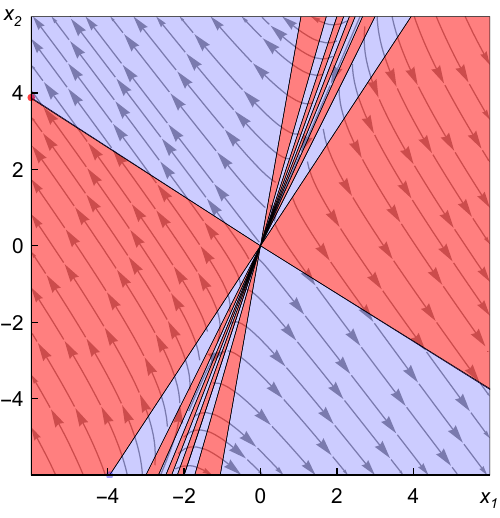}
    \end{subfigure}
    \qquad
    \begin{subfigure}[b]{0.28\textwidth}
        \includegraphics[width=\textwidth]{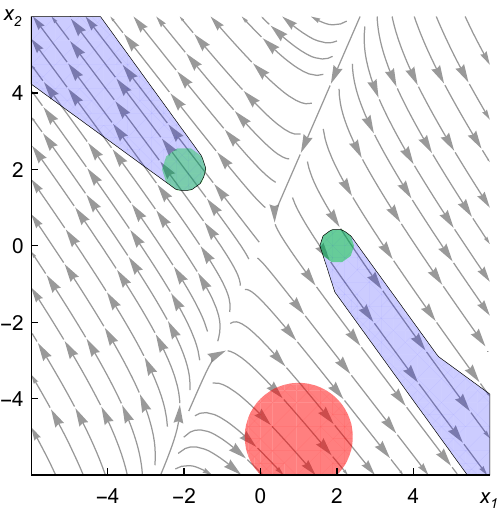}
    \end{subfigure}\\
     \begin{subfigure}[b]{0.2\textwidth}
     \begin{align*}
     &\,\quad x_1'=-2x_1 + x_2, \\
     &\,\quad x_2'=x_1 - 3x_2. \\
     &\text{(\rref{fig:ex1} right example)}\\~\\~\\~\\
     \end{align*}
    \end{subfigure}
    \quad\,
    \begin{subfigure}[b]{0.28\textwidth}
        \flushleft
        \includegraphics[width=\textwidth]{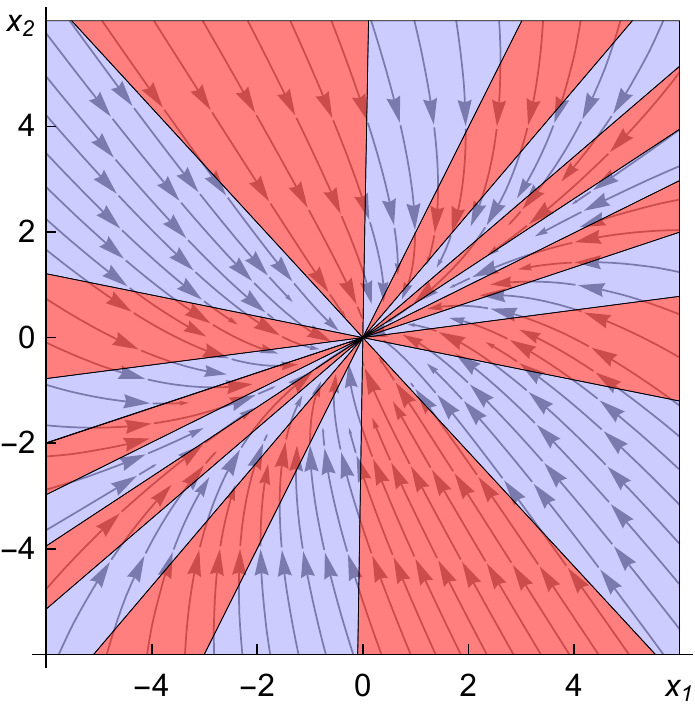}
    \end{subfigure}
    \qquad
    \begin{subfigure}[b]{0.28\textwidth}
        \includegraphics[width=\textwidth]{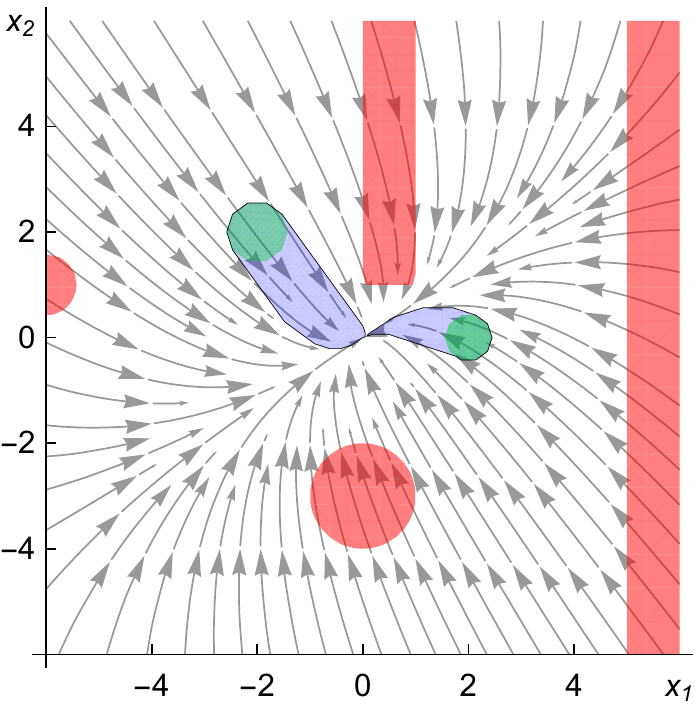}
    \end{subfigure}

	\caption{A visualization of our implementation of the conic abstractions method (each example is shown row-wise).
The left figures show the generated conic partition into $20$ cones (alternating red and blue colors).
The right figures show the reachable set computation (in blue) from the same green initial sets as in~\rref{fig:ex1}.
These reachable sets, which are invariant sets, suffice to show that the ODE never reaches any unsafe states (in red).
The method automatically produces finer partitions of the state space (using more cones) when the direction of the vector field changes more drastically.
The top partition concentrates several cones around its unstable manifold~\cite{Chicone2006,Strogatz2001} (the line $ y = \frac{1}{6} (1+\sqrt{13}) x $), while the bottom partition has more evenly spaced out cones.}
\label{fig:cabs}
\end{figure}

\subsection{Qualitative Analysis for Non-Linear Systems}
General non-linear polynomial systems of ODEs present a hard class of problems for invariant generation.
A number of useful heuristics can be applied to partition the continuous state space of these systems, in the hope that the resulting abstraction exhibits a suitable invariant.
For example, factorizing the RHS of a differential equation \(x_i'=f_i(x)\) yields a set of irreducible polynomial factors $p_1,\dots,p_k$ such that $f_i = \prod_{j=1}^k p_j$, which implies that the flow along the curves $p_j=0$ vanishes in the $x_i$ direction.
This information can be used to cheaply approximate the transition relation in the discrete abstraction and to efficiently extract \emph{invariant candidates}.
For the non-linear ODE in~\rref{fig:abstr}, the discretization polynomials $p_1,p_2$ are chosen such that $x_2'=0$ and $x_1'=0$ on their respective level curves.
This yields a useful discrete abstraction e.g. $S_4$ is an invariant for the resulting abstraction (\rref{fig:abstraction}).
Other useful sources of polynomials for qualitative analysis of non-linear systems
are found in, e.g. the summands and irreducible factors of the right-hand sides of the ODEs, the Lie derivatives
of the factors, and physically meaningful quantities such as the \emph{divergence}
of the system's vector field.

\begin{figure}[h]
    \centering
    \includegraphics[width=0.8\columnwidth,trim=0 11 5 11,clip]{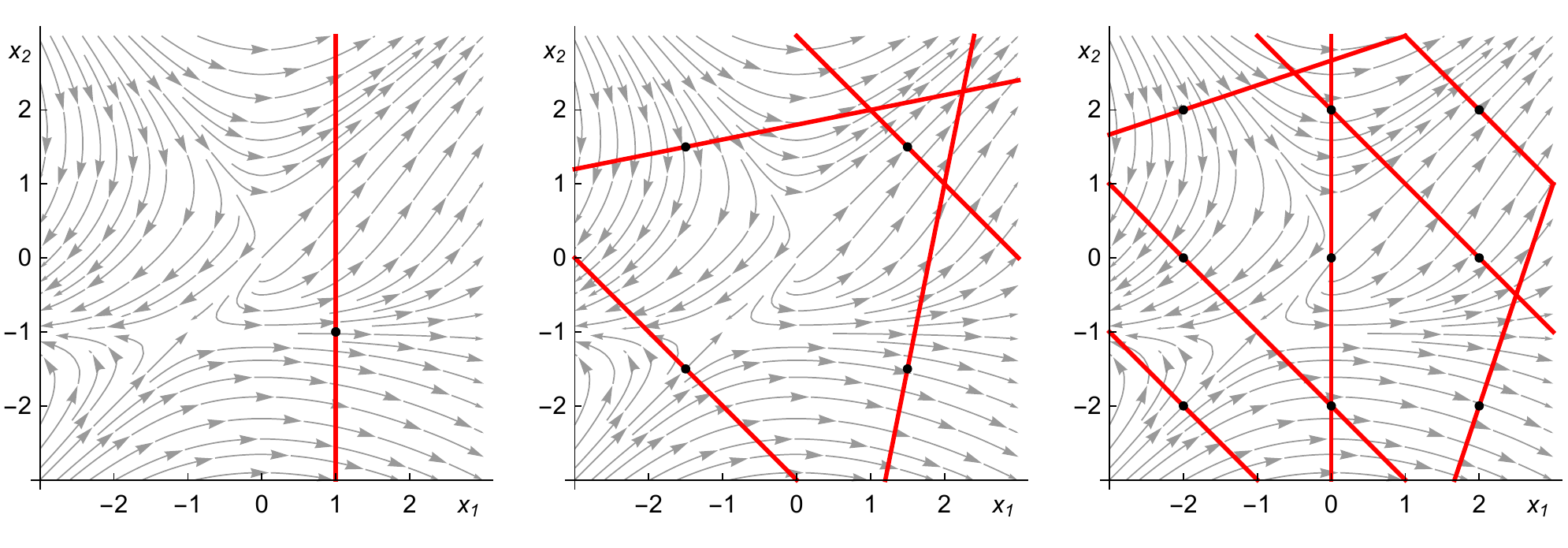}
    \caption{Abstractions using locally transverse linear forms (shown as red lines) generated from a grid of points (in black)}
    \label{fig:locorth}
\end{figure}

\paragraph{Locally transverse linear forms}
A simple geometric idea can sometimes help generate linear polynomials for abstraction. 
For a system of ODEs $\vec{x}'=f(\vec{x})$, which may be non-linear,
and a regular point $\vec{x}_0 \in \mathbb{R}^n$ with $f(\vec{x}_0)\neq \vec{0}$, one may construct the linear
form $f(\vec{x}_0)\cdot(\vec{x}-\vec{x}_0)$, which has the property that its
zero set is locally \emph{transverse} to the vector field near $\vec{x}_0$.\footnote{
By continuity of $f(\cdot)$, the vectors $f(\vec{x})$ are sufficiently close to $f(\vec{x}_0)$ for points $\vec{x}$ in a small neighborhood around $\vec{x}_0$. Therefore, all ODE solutions in this neighborhood can only cross $f(\vec{x}_0)\cdot(\vec{x}-\vec{x}_0)$ in the same direction as $f(\vec{x}_0)$.}
With a sufficiently fine partitioning using regular points, one has a good chance of finding invariant regions in the abstraction.
In problems where the evolution domain constraint describes a bounded set, it
is possible to obtain useful abstractions by choosing a finite number of regular
points $\vec{x}_0$ within the set and partitioning the constraint with the corresponding
locally transverse linear forms (as illustrated in~\rref{fig:locorth}). Of
course, choosing ``good'' points is the main problem in this method; one
possibility is to use evenly-spaced points forming a grid covering the evolution domain constraint.

\subsection{General-Purpose Methods}
\label{sec:generic}

Beyond qualitative analysis, Pegasus implements several general-purpose
invariant generation techniques which represent \emph{restricted, but tractable fragments}
of the general method of template enumeration. The search for symbolic parameters
in these methods is \emph{not} performed using real quantifier elimination, but
instead takes place in more tractable theories.

\subsubsection{Polynomial First Integrals}
\label{subsec:polyFI}

A polynomial $p \in \mathbb{R}[\vec{x}]$ is a \emph{first integral}~\cite[2.4.1]{Goriely2001} (also see~\cite[\S 23]{Pontryagin1962}) of the system \mbox{$\D{\vec{x}} = f(\vec{x})$} iff its Lie derivative $\D{p}$ with respect to the vector field $f$ is the zero polynomial.
First integrals are also known as \emph{conserved quantities} because they have an important property: their value never changes along the solutions to ODEs; that is to say, for any $k\in\mathbb{R}$, $p=k$ is an invariant of the system.\footnote{Strictly speaking, first integrals and conserved quantities are \emph{not the same}: a first integral may only be considered a conserved quantity in regions where it is defined. In this case, however, polynomial functions are defined everywhere in $\mathbb{R}^n$ and the two notions coincide.}
For a single first integral $p$, if one were to use (the signs of) the polynomial $p-k$ to build an abstraction, the abstract state space would not feature any transitions between its states (illustrated in \rref{fig:FI}).
Thus, one has the freedom to choose values $k$ for which the resulting discrete abstraction suitably partitions the state space.
For example, if the initial states lie entirely within $p < k$ and the unsafe ones within $p > k$, then $p < k$ is an invariant separating those sets.

\begin{figure}[h!]
\centering

\begin{tikzpicture}[transform shape, scale=1,node distance=3cm,state node/.style={circle,draw,font=\sffamily\normalsize\bfseries}, scale=0.75]
   \node[state node] (M1) [inner sep=2pt] {$p<k$};
   \node[state node] (M2) [right of=M1, inner sep=2pt] {$p=k$};
   \node[state node] (M3) [right of=M2, inner sep=2pt] {$p>k$};
\end{tikzpicture}
\caption{Discrete abstraction with first integral $p-k$ ($k\in\mathbb{R}$)}
\label{fig:FI}
\end{figure}
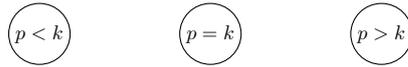

Pegasus can search for \emph{all} polynomial first integrals up to a configurable degree bound by solving a system of \emph{linear equations} whose solutions provide the coefficients of the bounded degree polynomial template for the first integral.
This is known as the \emph{method of undetermined coefficients}; we illustrate the main steps of the method in the following example.

\begin{example}[Kasner's equations]
	\label{ex:kasner}
	Consider the non-linear system of ODEs describing a special case of Einstein's
	gravitational equations~\cite{kasner1925}
	\begin{align*}
		x_1'&=x_2x_3 -x_1^2\,, \\
		x_2'&=x_3x_1 - x_2^2\,,\\
		x_3'&=x_1x_2-x_3^2\,,
	\end{align*}
and a polynomial template of maximum degree $2$ in the state variables $x_1,x_2,x_3$:
    \[
	    p_{\textcolor{red}{\vec{a}},2} =
	    \textcolor{red}{a_0}
	    + \textcolor{red}{a_1}x_1
	    + \textcolor{red}{a_2}x_2
	    + \textcolor{red}{a_3}x_3
	    + \textcolor{red}{a_4}x_1^2
	    + \textcolor{red}{a_5}x_1x_2
	    + \textcolor{red}{a_6}x_1x_3
	    + \textcolor{red}{a_7}x_2^2
	    + \textcolor{red}{a_8}x_2x_3
	    + \textcolor{red}{a_9}x_3^2
	    \enspace.
    \]
	Computing the Lie derivative of this template with respect to the system, i.e. $(p_{\textcolor{red}{\vec{a}},2})' = \frac{\partial p_{\textcolor{red}{\vec{a}},2}}{\partial x_1} x_1' + \frac{\partial p_{\textcolor{red}{\vec{a}},2}}{\partial x_2} x_2'+ \frac{\partial p_{\textcolor{red}{\vec{a}},2}}{\partial x_3} x_3'$ gives a degree $3$ parametric polynomial:
\begin{align*}
	(p_{\textcolor{red}{\vec{a}},2})' &=
	- \textcolor{red}{a_1}x_1^2
	+ \textcolor{red}{a_3}x_1x_2
	+ \textcolor{red}{a_2}x_1x_3
	- \textcolor{red}{a_2}x_2^2
	+ \textcolor{red}{a_1}x_2x_3
	- \textcolor{red}{a_3}x_3^2
	- 2\textcolor{red}{a_4}x_1^3 \\
	&\quad
	+ (\textcolor{red}{a_6}- \textcolor{red}{a_5})x_1^2x_2
	+ (\textcolor{red}{a_5}- \textcolor{red}{a_6})x_1^2x_3
	+ (\textcolor{red}{a_8}- \textcolor{red}{a_5})x_1x_2^2\\
	&\quad
        + (2\textcolor{red}{a_4} + 2\textcolor{red}{a_7} + 2\textcolor{red}{a_9})x_1x_2x_3
	+ (\textcolor{red}{a_8}- \textcolor{red}{a_6})x_1x_3^2
	- 2\textcolor{red}{a_7}x_2^3 \\
	&\quad
	+ (\textcolor{red}{a_5}- \textcolor{red}{a_8})x_2^2x_3
	+ (\textcolor{red}{a_6}- \textcolor{red}{a_8})x_2x_3^2
	- 2\textcolor{red}{a_9}x_3^3\enspace.
\end{align*}
In order to find a first integral, one is required to solve the equation $(p_{\textcolor{red}{\vec{a}},2})' = 0$, but a polynomial is $0$ precisely when all of its coefficients are $0$.
Thus, by equating all coefficients of the Lie derivative to $0$, finding a first integral reduces to solving a \emph{linear system of equations} over the symbolic coefficients $\textcolor{red}{a_0}$, \dots ,$\textcolor{red}{a_9}$:
\begin{align*}
	&{-}\textcolor{red}{a_1}=0,
	 \textcolor{red}{a_3}=0,
	 \textcolor{red}{a_2}=0,
	{-}\textcolor{red}{a_2}=0,
	 \textcolor{red}{a_1}=0,
	{-}\textcolor{red}{a_3}=0,
	{-}2\textcolor{red}{a_4}=0,
	 (\textcolor{red}{a_6}- \textcolor{red}{a_5})=0, \\
	&\quad
	 (\textcolor{red}{a_5}- \textcolor{red}{a_6})=0,
	 (\textcolor{red}{a_8}- \textcolor{red}{a_5})=0,
	 (2\textcolor{red}{a_4}+ 2\textcolor{red}{a_7}+ 2\textcolor{red}{a_9})=0,
	 (\textcolor{red}{a_8}- \textcolor{red}{a_6})=0, \\
	&\quad
	{-}2\textcolor{red}{a_7}=0,
	 (\textcolor{red}{a_5}- \textcolor{red}{a_8})=0,
	 (\textcolor{red}{a_6}- \textcolor{red}{a_8})=0,
	{-}2\textcolor{red}{a_9}=0
	\enspace.
\end{align*}
Solutions are efficiently found using linear algebra~\cite[\S 2.4.1]{Goriely2001}.
In this example, a non-trivial solution yields the polynomial first integral $x_1x_2+x_1x_3+x_2x_3$.
Moreover, \emph{all} first integrals of degree (up to) two provide concrete instances of the coefficients $\textcolor{red}{a}$ and so must correspond to a solution of these equations.
\end{example}

When a polynomial first integral $p$ is computed, one has the freedom of choosing its initial value, which is guaranteed to remain invariant throughout the evolution of the system. In the above example, one may choose any real number $k$ and partition the state space into invariant regions defined by the sign conditions on the polynomial $x_1x_2+x_1x_3+x_2x_3 -k$. To obtain a tight over-approximation of the reachable set from the initial set of states given in the verification problem, one may choose $k$ by maximizing and minimizing the value of the first integral $p$ on the initial set of states within the evolution domain constraint, i.e., one may search for the real values (if they exist):
\[
		k_{\max} = \max_{\vec{x} \in \Init \cap Q}~ p(\vec{x})\,,
	\qquad
		k_{\min} = \min_{\vec{x} \in \Init \cap Q}~ p(\vec{x})\enspace.
\]

If finite values $k_{\max}$ and $k_{\min}$ can be obtained, one may generate a continuous invariant $k_{\min}\leq p \land p \leq k_{\max}$ (or just $p=k_{\min}$ if $k_{\max} = k_{\min}$).

\begin{dangerous}
	Maximizing/minimizing multivariate polynomials subject
	to semi-algebraic constraints often leads to irrational
	and \emph{real algebraic numbers} as exact maxima/minima. Numerical algorithms will
	yield values that are near-optimal, which may require them to be increased/decreased
	by some amount before a genuine invariant is constructed as described above.
\end{dangerous}

\begin{dangerous}
	The set $\Init \cap Q$ may have \emph{multiple connected components}, and tighter invariants may be obtained from first integrals when the value $k$ is optimized subject to each connected component separately.
  A cheap way to approximate the connected components is to normalize $\Init \land Q$ to disjunctive normal form and consider each disjunct as a separate component.
\end{dangerous}

If more than one independent first integral for a system is found, one may construct
finer abstractions and generate tighter invariants over-approximating the reachable set. A particularly
interesting case is when an $n$-dimensional system of ODEs has $n-1$ \emph{functionally independent}
algebraic first integrals: such a system is said to be \emph{algebraically integrable}~\cite{Goriely2001,olver2000}.
In such a system, given a state $\vec{x}_0 \in \mathbb{R}^n$, one may evaluate the
first integrals $p_1,p_2,\dots,p_{n-1}$ at that state to obtain a continuous invariant given by
\[
	p_1=p_1(\vec{x}_0) \land p_2=p_2(\vec{x}_0) \land \cdots \land p_{n-1}=p_{n-1}(\vec{x}_0)\enspace.
\]
If the first integrals are functionally independent, i.e. when the matrix
\[ [\nabla p_1~\nabla p_2~\cdots~\nabla p_{n-1}]\]
whose columns are formed by the gradients $\nabla p_i \equiv \left(\frac{\partial p_i}{\partial x_1}, \frac{\partial p_i}{\partial x_2}, \dots, \frac{\partial p_i}{\partial x_n}\right)^T$ has \emph{full rank} at $\vec{x}_0$ (i.e. when the vectors $\nabla p_i$ evaluated at $\vec{x}_0$ are linearly independent, see e.g.~\cite{olver2000}), the resulting conjunctive formula (locally) describes a $1$-dimensional invariant curve in $n$-dimensional state space and provides the tightest possible algebraic invariant containing $\vec{x}_0$.

\begin{example}[Algebraic integrability]
	\begin{figure}[ht]
\centering
\begin{subfigure}[b]{.45\textwidth}
\includegraphics[width=\textwidth]{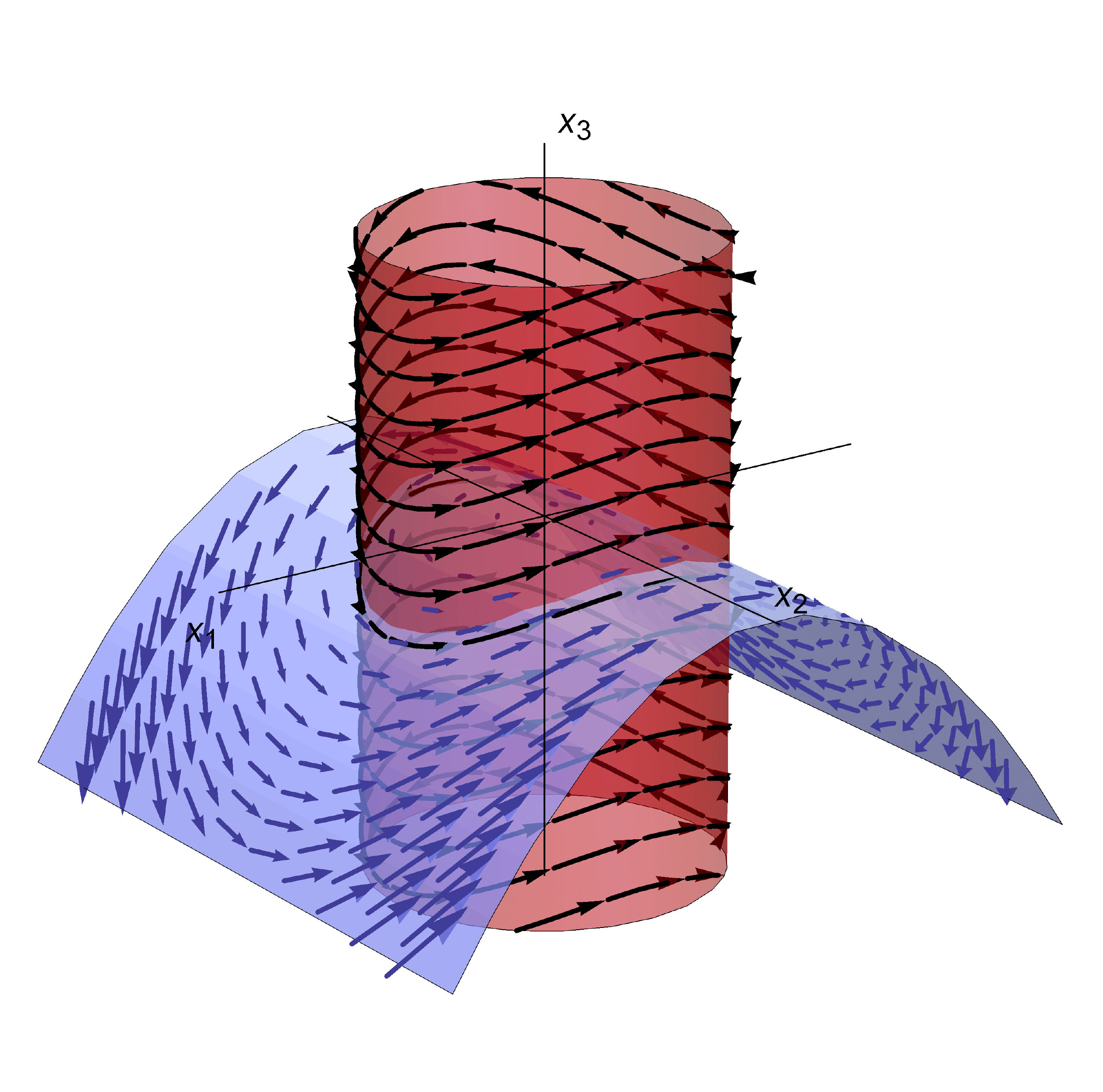}
	\caption{\label{subfig:algint1} Invariant surfaces $p_1=1$ and $p_2=0$}
\end{subfigure}
\quad
\begin{subfigure}[b]{.45\textwidth}
\includegraphics[width=\textwidth]{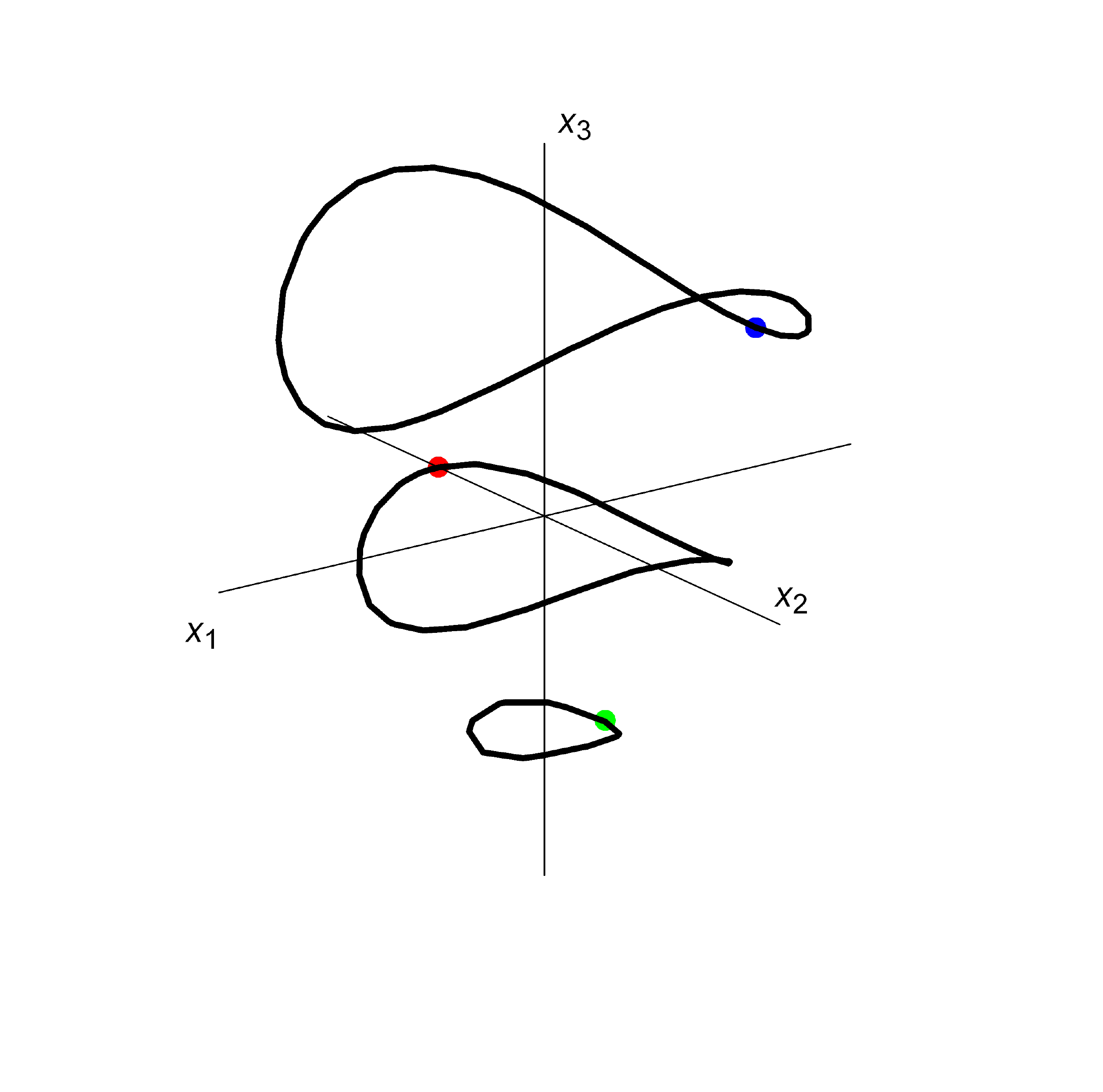}
	\caption{\label{subfig:algint2} Invariant curves through points}
\end{subfigure}
\vspace{5pt}
\caption{ Invariant level sets of two independent first integrals (left) whose intersections define invariant curves (right)}
\label{fig:AlgInt}
\end{figure}
 Consider the non-linear system
\begin{align*}
x_1' &= -x_2\,, \\
x_2' &= x_1\,, \\
x_3' &= x_1x_2\enspace.
\end{align*}
	Using a quadratic polynomial template $p_{\textcolor{red}{\vec{a}},2}$ and solving the linear system of equations corresponding to the equality $(p_{\textcolor{red}{\vec{a}},2})'=0$ as described in Example~\ref{ex:kasner}, one obtains the first integrals $p_1=x_1^2+x_2^2$ and $p_2=x_1^2+x_3$. The level sets described by $p_1=k_1$ and $p_2=k_2$ are invariants for any $k_1,k_2\in \mathbb{R}$. A level set of a first integral corresponds to an invariant surface to which the system's vector field is tangent at all points on the surface. For example, Fig.~\ref{subfig:algint1} illustrates two invariant surfaces of this system, which are described by $p_1=1$ (corresponding to the red cylinder) and $p_2=0$ (corresponding to the blue inverted parabolic surface).
	Taking $\vec{x}_0=(0,1,0)^T$, one can easily check that the first integrals $p_1$ and $p_2$ are functionally independent:
\[
	[\nabla p_1~\nabla p_2] =
\begin{bmatrix}
\frac{\partial p_1}{x_1} & \frac{\partial p_2}{x_1} \\
\frac{\partial p_1}{x_2} & \frac{\partial p_2}{x_2} \\
\frac{\partial p_1}{x_3} & \frac{\partial p_2}{x_3}
\end{bmatrix}
=
\begin{bmatrix}
2x_1 & 2x_1 \\
2x_2 & 0 \\
0 & 1
\end{bmatrix}
	\,, \text{ which at $\vec{x}_0$ becomes }~
\begin{bmatrix}
0 & 0 \\
2 & 0 \\
0 & 1 \\
\end{bmatrix}\,
\]
	and is full rank. Since the system of ODEs is $3$-dimensional and we have $2=3-1$ independent algebraic first integrals, this system is algebraically integrable.\footnote{In this example the first integrals are polynomial functions, but in general algebraic first integrals need \emph{not} be polynomial: e.g. they may be rational functions, as we shall see in Sec.~\ref{subsec:ratFI}.}
	Intuitively, the invariant level surfaces of first integrals will intersect transversally (i.e. will not be tangent) if the first integrals are functionally independent. Each such intersection results in an invariant which is of lower dimension: for example, the intersection of the two invariant surfaces in Fig.~\ref{subfig:algint1} (i.e. $p_1=1 \land p_2=0$) corresponds to the invariant \emph{space curve} -- a one-dimensional object in $3$-dimensional space -- which contains the point $\vec{x}_0 = (0,1,0)^T$, as illustrated in Fig.~\ref{subfig:algint2} by the middle curve going through the red point $\vec{x}_0$.\footnote{In fact, for this particular example this closed curve represents the \emph{periodic orbit} (see e.g.~\cite{Chicone2006}) of the system through the point $\vec{x}_0$.} One may choose other points $\vec{x}_0$ and use them to evaluate the first integrals $p_1(\vec{x}_0)$ and $p_2(\vec{x}_0)$, from which one can construct other invariant curves described by $p_1 = p_1(\vec{x}_0) \land p_2 = p_2(\vec{x}_0)$ (as in Fig.~\ref{subfig:algint2}).
\end{example}

\subsubsection{Darboux Polynomials}
\label{subsec:darboux}

Darboux polynomials were first introduced in 1878~\cite{Darboux} to study integrability of polynomial ODEs.
A polynomial $p \in \mathbb{R}[\vec{x}]$ is said to be a \emph{Darboux polynomial} for the system $\D{\vec{x}} = f(\vec{x})$ if and only if \mbox{$\D{p} = \alpha p$} for some polynomial $\alpha \in \mathbb{R}[\vec{x}]$, which is known as the \emph{cofactor} of $p$.
Like first integrals, discrete abstractions produced with Darboux polynomials result in three states with no transitions between them (as illustrated in \rref{fig:FI}, but with $k=0$).
Unlike first integrals, only $p=0$ is guaranteed to be an invariant
of the system.  Darboux polynomials have been used for predicate abstraction of continuous
systems by Zaki \emph{et al.}~\cite{DBLP:journals/jacic/ZakiDTB09}, who successfully
applied them to verify electrical circuit designs.

The problem of generating Darboux polynomials is generally far more difficult than that of generating polynomial first integrals (which represent the special case of Darboux polynomials where the cofactor $\alpha$ is $0$ in the equation $p'=\alpha p$). A modification of the method of undetermined coefficients described in the previous section can likewise be applied to search for Darboux polynomials. However, in order to apply this method, one is required to provide a polynomial template for \emph{both} the Darboux polynomial \emph{and} for its cofactor. Whenever one has a polynomial system of ODEs $\vec{x}'=f(\vec{x})$ in which the maximum polynomial degree of the components $f_1,f_2,\dots,f_n$ of $f$ is some $r\geq 0$, then the maximum possible degree of the Lie derivative (w.r.t. this system) of a polynomial $p$ of maximum degree $d$  is given by $d + r -1$. Consequently, to search for a Darboux polynomial of maximum degree $d$, the maximum degree of the cofactor $\alpha$ in the equation $p'=\alpha p$ that one needs to consider is given by $r-1$. To apply the method of undetermined coefficients, one requires a template $p_{\textcolor{red}{\vec{a}},d}$ for the Darboux polynomial and a separate template $\alpha_{\textcolor{red}{\vec{b}},r-1}$ for the cofactor. The equation to be solved is the following:
\[
	(p_{\textcolor{red}{\vec{a}},d})' - \alpha_{\textcolor{red}{\vec{b}},r-1} p_{\textcolor{red}{\vec{a}},d} =0 \enspace.
\]
By expanding the polynomial on the left-hand side and equating each of its monomial coefficients to $0$, one obtains a system of equations in the symbolic parameters $\textcolor{red}{\vec{a}},\textcolor{red}{\vec{b}}$; however, while this system is linear in the parameter variables $\textcolor{red}{\vec{a}}$ and $\textcolor{red}{\vec{b}}$ considered separately, it is a \emph{non-linear system of equations} in $\textcolor{red}{\vec{a}},\textcolor{red}{\vec{b}}$ simultaneously. In practice, solving such a non-linear system is far more computationally expensive than solving the linear systems for polynomial first integrals; the na\"ive method of undetermined coefficients does not provide a practically appealing solution for Darboux polynomial generation.

Fortunately, automatic generation of Darboux polynomials is an active area of research, owing largely to their importance as a crucial component in the \emph{Prelle-Singer method}~\cite{prelle1983} for computing elementary closed-form solutions to ODEs.
In order to implement the Prelle-Singer method, more sophisticated algorithms for Darboux polynomial generation have been developed in the computer algebra community, e.g.  two algorithms were reported by Man~\cite{man1993}. Indeed, in our experiments we have found the algorithms \texttt{ps\_1} and \texttt{new\_ps\_1} in~\cite{man1993} to be much more practical and implement them in Pegasus.

\begin{remark}
	We remark also that several algorithms for generating (what are essentially) Darboux polynomials have more recently been developed within the verification community~\cite{DBLP:conf/hybrid/KongBSJH17,DBLP:journals/tcs/RebihaMM15,DBLP:journals/fmsd/SankaranarayananSM08}. However, our experience with some of these procedures has been less positive.
The method in~\cite{DBLP:journals/tcs/RebihaMM15} was in practice found to be very inefficient and \emph{incomplete}, i.e. unable in general to find all the Darboux polynomials matching a given polynomial template;  the technique described in~\cite{DBLP:conf/hybrid/KongBSJH17} is significantly faster but is likewise incomplete.
\end{remark}

\begin{dangerous}
	Determining whether an arbitrary system of polynomial ODEs possesses a Darboux polynomial (and finding a bound on its degree if it does) remains an open problem~\cite[\S 4.1]{zhang2017}.
\end{dangerous}

\subsubsection{Rational First Integrals}
\label{subsec:ratFI}

Beyond polynomial functions, a much larger class of algebraic conserved quantities is that of \emph{rational first integrals}; these are first integrals represented by \emph{rational functions}, i.e. functions of the form $\frac{a}{b}$, where $a,b$ are polynomials and $b \neq 0$. Searching for this kind of first integral is (unsurprisingly) more difficult than is the case with polynomials; however, it is made possible by exploiting an idea from the seminal work of Darboux (see e.g. Schlomiuk~\cite{Schlomiuk1993}): multiple Darboux polynomials can be combined to construct a rational first integral.

\begin{theorem}
\label{thm:dbxthm}%
	Let $p_1,p_2, \dots, p_k$ be Darboux polynomials for the system \m{\vec{x}'=f(\vec{x})}, with \mbox{$p_i'=\alpha_ip_i$}, where $\alpha_i$ is some polynomial cofactor for each $i=1,\dots,k$. If
\begin{equation}
\lambda_1\alpha_1 + \lambda_2\alpha_2 + \dots + \lambda_k\alpha_k=0
\label{eq:dbxthm}
\end{equation}
 has a non-trivial integer solution, i.e. $\vec{\lambda}=(\lambda_1,\lambda_2,\dots,\lambda_k) \in \integers^k \setminus \{ \vec{0} \}$, then the system has a rational first integral $r_{\vec{\lambda}} \in \mathbb{R}(\vec{x})$ given by the product
 \[
r_{\vec{\lambda}} = p_1^{\lambda_1}p_2^{\lambda_2}\cdots p_k^{\lambda_k}.
 \]
\end{theorem}
\begin{proof}
By applying the product rule to compute the Lie derivative $r_{\vec{\lambda}}'$, we get
\begin{align*}
(p_1^{\lambda_1}p_2^{\lambda_2}\cdots p_k^{\lambda_k})' & =\lambda_1p_1^{\lambda_1-1}p_1'(p_2^{\lambda_2}\cdots p_k^{\lambda_k}) + \dots + \lambda_k p_k^{\lambda_k -1}p_k'(p_1^{\lambda_1}\cdots p_{k-1}^{\lambda_{k-1}})\\
&=\lambda_1p_1^{\lambda_1-1}\alpha_1 p_1(p_2^{\lambda_2}\cdots p_k^{\lambda_k}) + \dots + \lambda_k p_k^{\lambda_k -1}\alpha_k p_k(p_1^{\lambda_1}\cdots p_{k-1}^{\lambda_{k-1}})\\
& = (\lambda_1\alpha_1 + \lambda_2\alpha_2 + \dots + \lambda_k\alpha_k)(p_1^{\lambda_1}p_2^{\lambda_2}\cdots p_k^{\lambda_k}).
\end{align*}
From equation~\rref{eq:dbxthm} it follows that $r_{\vec{\lambda}}'=0$ and $r_{\vec{\lambda}}$ is therefore a first integral.
\qed
\end{proof}
\begin{remark}
Obviously, if the solution to~\rref{eq:dbxthm} is such that $\vec{\lambda} \in {\integers}_{\geq 0}^k$, then the first integral is polynomial; at least one negative component in $\vec{\lambda}$ is therefore required in order to construct a non-polynomial rational first integral.
	We also note that one may search for rational solutions to~\rref{eq:dbxthm}, i.e. $\vec{\lambda} \in \rationals^k$, which will in general result in first integrals featuring radicals. Any such first integral can be turned into a rational first integral by raising it to an integer power corresponding to the least common multiple of the denominators of the rational numbers $\lambda_1,\dots,\lambda_k$.
	In general, $\lambda_1,\dots,\lambda_k$ need not be rational or even real numbers in order for the construction given in \rref{thm:dbxthm} to work; however, irrational solutions lead to first integrals that are not rational functions.
\end{remark}

In light of the above theorem, a straightforward procedure for generating rational first integrals (which has previously been suggested by Man~\cite{Man1994}) involves \begin{enumerate*}[label=(\roman*)] \item generating Darboux polynomials $p_1,p_2\dots,p_k$ for the system $\vec{x}'=f(\vec{x})$, e.g. using an implementation of Man's algorithms~\cite{man1993}, and \item finding integer (or rational) solutions to the linear system of equations~\rref{eq:dbxthm} in~\rref{thm:dbxthm}.\end{enumerate*} If the coefficients of the cofactors $\alpha_1,\alpha_2,\dots,\alpha_k$ in equation~\rref{eq:dbxthm} are all rational numbers, the problem reduces to solving a system of linear Diophantine equations, for which there exist polynomial-time algorithms.
If a rational first integral $r_{\vec{\lambda}} = \frac{a}{b}$ is found, then $\frac{a}{b}=l$ defines an invariant hypersurface for any choice of $l\in \mathbb{R}$, assuming $b\neq 0$; rewriting this, we get that $a-lb=0$ is invariant for any $l\in\reals$ (when $b\neq 0$).

\begin{example}
	\begin{figure}[ht]
\includegraphics[width=\columnwidth]{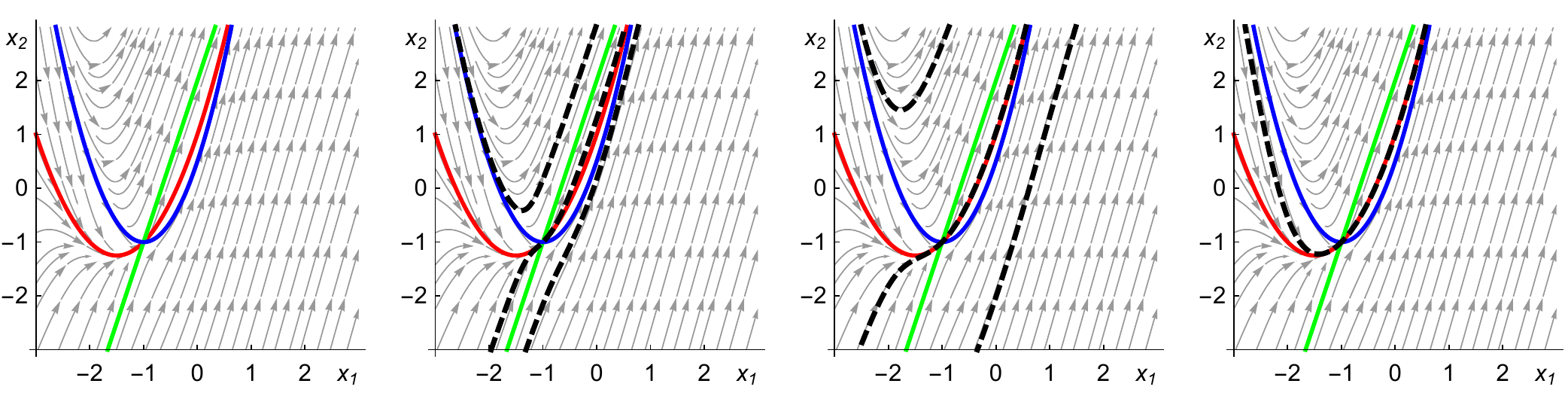}
\begin{subfigure}[t]{0.23\textwidth}
	\caption{$p_1,p_2,p_3 = 0$ \label{subfig:darboux}}
\end{subfigure}
\begin{subfigure}[t]{0.25\textwidth}
	\caption{$r_{\vec{\lambda}}=\frac{1}{10}$}
\end{subfigure}
\begin{subfigure}[t]{0.25\textwidth}
	\caption{$r_{\vec{\lambda}}=1$}
\end{subfigure}
\begin{subfigure}[t]{0.25\textwidth}
	\caption{$r_{\vec{\lambda}}=-2$}
\end{subfigure}
\caption{
	Rational first integral $r_{\vec{\lambda}}$ constructed from three
	Darboux polynomials. Zero sets of the three Darboux polynomials shown
	in solid green, blue and red. Invariant level sets of the rational
	first integral shown in dashed black for values
	$r_{\vec{\lambda}}=\frac{1}{10},1,-2$, respectively
	}
\label{fig:ratfi}
\end{figure}
	Consider the following non-linear system of ODEs~\cite{ferragut2010}:
\begin{align*}
x_1' &= 6 x_1^4+27 x_1^3-9 x_1^2 x_2+42 x_1^2-24 x_1 x_2+21 x_1+4 x_2^2-7 x_2+4\,, \\
x_2' &= 18 x_1^4+99 x_1^3-39 x_1^2 x_2+150 x_1^2+2 x_1 x_2^2-80 x_1 x_2+71 x_1+12 x_2^2-21 x_2+12\enspace.
\end{align*}

	Using our implementation of Man's algorithm~\cite{man1993}, we obtain the following list of Darboux polynomials in under one second of computation time:
\[(p_1,p_2,p_3) = \left(~x_1-\frac{x_2}{3}+\frac{2}{3},~x_1^2+2 x_1-\frac{2 x_2}{3}+\frac{1}{3},~x_1^2+3 x_1-x_2+1~\right)\enspace. \]

	Solving equation~\rref{eq:dbxthm} in~\rref{thm:dbxthm}, we obtain the solution $(\lambda_1,\lambda_2,\lambda_3)=(2,1,-1)$, from which we obtain the rational first integral (illustrated in~\rref{fig:ratfi})

		\[r_{\vec{\lambda}} = p_1^{2} p_2^{1} p_3^{-1} =  \frac{(x_1-\frac{x_2}{3}+\frac{2}{3})^2(x_1^2+2x_1-\frac{2 x_2}{3} + \frac{1}{3})}{x_1^2+3x_1-x_2+1}\enspace.\]

\end{example}

\begin{remark}
	Before attempting to search for algebraic first integrals (whether polynomials or rational functions) it is helpful to have static criteria that determine whether such first integrals can arise in a given system of ODEs. Criteria for non-existence of various kinds of first integrals have been studied by numerous authors (notably by Poincar\'e~\cite[\S 7.2]{zhang2017}) and typically make use of the linearization $\vec{x}'=A\vec{x}$ of the system $\vec{x}'=f(\vec{x})$ around a point of equilibrium (i.e. a point $\vec{x}_*$ where $f(\vec{x}_*)=\vec{0}$). In particular, a sufficient criterion for non-existence of rational first integrals in non-linear systems of ODEs  is given by Shi~\cite[Theorem 1]{shi2007nonexistence}; it requires that the eigenvalues $\lambda_1, \dots, \lambda_n$ of the matrix $A$ are such that $k_1\lambda_1 + \cdots + k_n\lambda_n =0$ does not have a non-trivial integer solution $(k_1,\dots,k_n) \in \mathbb{Z}^n \setminus \{\vec{0}\}$. A similar criterion, which furthermore accounts for repeated eigenvalues, is given by Goriely~\cite[Ch. 5, Prop. 5.5]{Goriely2001}.
\end{remark}

\paragraph{Combining Darboux Polynomials and Rational First Integrals.}
As a first hint of its flexibility for combining invariant generation methods, Pegasus implements rational first integral generation by combining several ideas described thus far in~\rref{sec:pegasus} as follows.
This flexibility is further exploited in the discussion of \emph{strategies} in~\rref{sec:saturation}.

\begin{enumerate}
	\item Compute a list of Darboux polynomials $p_1,\dots,p_k$  of some maximum polynomial degree $d$ using generation methods from~\rref{subsec:darboux}.
	\item Abstract the state space into sign invariant cells using those polynomials, e.g., $S_1 \equiv p_1<0 \land p_2=0$, $S_2 \equiv p_1<0 \land p_2>0$, $S_3 \equiv p_1 < 0 \land p_2 < 0$, etc., as described in~\rref{sec:ExactDiscreteAbstraction}. Notably, the resulting abstraction has no transitions between its discrete states, as illustrated in~\rref{fig:FI}.
  \item Prune away those invariant cells that do not intersect the initial set of states, e.g., delete $S_1$ if $\Init \cap S_1 = \emptyset$  since $S_1$ is then unreachable. Similarly, prune away cells that do not intersect the unsafe set, e.g., delete $S_2$ if $\Unsafe \cap S_2 = \emptyset$ because no initial states in $S_2$ can reach the unsafe set.
	\item \label{step4ratfi} The remaining unpruned \emph{conflict cells}, say $S_3$, define new invariant generation \emph{sub-problems}, where the original evolution domain constraint $Q$ is restricted to $Q \land S_3$.
Each of the Darboux polynomials are sign-invariant in these cells; moreover, those Darboux polynomials that are sign-definite (either strictly positive or negative) in each cell, e.g. $p_1,p_2$ with evolution domain constraint $p_1 < 0 \land p_2 > 0$ for $S_3$, can be used to compute rational first integrals $r_{\vec{\lambda}}$ (following~\rref{thm:dbxthm}). The denominator of $r_{\vec{\lambda}}$ is guaranteed to be a product of (powers of) sign-definite polynomials so these rational functions are always defined within each conflict cell.
  \item \label{step5ratfi} Using their respective rational first integrals $r_{\vec{\lambda}}$, refine each conflict cell by maximizing and minimizing the values of $r_{\vec{\lambda}}$ to obtain invariant sub-level sets $k_{\min}\leq r_{\vec{\lambda}} \land r_{\vec{\lambda}} \leq k_{\max}$ over the initial set (restricted to that cell), as described in~\rref{subsec:polyFI}.
  \item If conflict cells remain, increase the polynomial degree $d$ and go to step \emph{1.}
\end{enumerate}

\paragraph{Rational First Integrals of Linear Systems.} In the case of linear
systems of ODEs \(\vec{x}'=A\vec{x}\), more efficient methods exist that allow us
to \emph{directly construct} rational first integrals from the eigenvalues and
eigenvectors of the system matrix $A$. These explicit constructions are
described, e.g. in the work of Gorbuzov \& Pranevich~\cite{gorbuzov2012} and
Falconi \& Llibre~\cite{falconi2004}; in Pegasus, we implement and deploy the former techniques~\cite{gorbuzov2012}.

It is instructive to compare the results obtained by Lafferriere, Pappas and
Yovine~\cite{DBLP:journals/jsc/LafferrierePY01} (which state that
semi-algebraic reachable sets of linear ODEs \(\vec{x}'=A\vec{x}\) can be constructed from
semi-algebraic initial sets in cases when $A$ is
diagonalizable and all of its eigenvalues are rational) to essentially analogous results
independently obtained in the study of integrability of linear systems. For instance, \cite[Property 1.1]{gorbuzov2012} gives a sufficient condition for algebraic integrability which states that a linear system
\(\vec{x}'=A\vec{x}\) has a \emph{basis} of rational first integrals (i.e. is algebraically integrable)
if all the eigenvalues of $A$ are rational and of multiplicity $1$. Indeed, such a
basis of rational first integrals enables one to construct reachable sets
described by polynomials.

      ~ %

\subsubsection{Barrier Certificates}
The method of \emph{barrier certificates} is a popular Lyapunov-like technique for safety verification of continuous and hybrid systems~\cite{DBLP:conf/hybrid/PrajnaJ04}.
Barrier certificates are differentiable functions $p$ that define an invariant region $p \leq 0$ which separates the initial states (wholly contained within \mbox{$p \leq 0$}) from the unsafe states (wholly contained within \mbox{$p > 0$}). In order to ensure continuous invariance of the region defined by $p \leq 0$, the Lie derivative $p'$ of the barrier certificate needs to satisfy certain criteria; differences in these criteria give rise to a number of variations of barrier certificates in the literature. The original work by Prajna and Jadbabaie~\cite{DBLP:conf/hybrid/PrajnaJ04} introduced \emph{convex barrier certificates}, which employ the differential inequality $p' \leq 0$ to guarantee invariance of $p \leq 0$ under the flow of the system. Later work by Kong et al.~\cite{DBLP:conf/cav/KongHSHG13} introduced so-called \emph{exponential-type barrier certificates}, which provide a generalization employing the differential inequality $p' \leq \lambda p$, where $\lambda \in \mathbb{R}$; this was generalized further yet in the work of Dai et al.~\cite{DBLP:journals/jsc/DaiGXZ17}, who introduced \emph{general barrier certificates} employing the differential inequality $p' \leq \omega(p)$, where $\omega$ is a specifically crafted scalar \emph{function} to guarantee invariance of $p \leq 0$. All of the above developments are fundamentally based on the classical notion of \emph{comparison systems}~\cite[Ch II, \S 3, Ch. IX]{Rouche1977} in the theory of ODEs. A unified understanding of these generalizations is described in prior work~\cite{DBLP:conf/fm/SogokonGTP18}, which introduces a further generalization of the barrier certificate framework: \emph{vector barrier certificates}, employing multidimensional comparison systems in a way analogous to vector Lyapunov functions introduced by Bellman~\cite{bellman1962vector}.

Barrier certificates are practically interesting because one may apply the method of undetermined coefficients to automatically search for them using tractable techniques: either sum-of-squares programming (SOS)~\cite{DBLP:conf/hybrid/PrajnaJ04} or linear programming (LP)~\cite{DBLP:conf/fm/YangHCL016}.
Pegasus is able to search for convex~\cite{DBLP:conf/hybrid/PrajnaJ04}, exponential-type~\cite{DBLP:conf/cav/KongHSHG13}, and vector barrier certificates~\cite{DBLP:conf/fm/SogokonGTP18} using both SOS and LP techniques.
However, the resulting barrier certificates often suffer from numerical inaccuracies arising from the use of semidefinite solvers and interior point methods~\cite{DBLP:journals/fmsd/RouxVS18}. Pegasus currently uses a simple rounding heuristic on the numerical result and explicitly checks invariance for the resulting (exact) barrier certificate candidates using real quantifier elimination.
An example of a barrier certificate generation technique implemented in Pegasus, and an illustration of its numerical issues is given next.

\begin{example}
Consider the safety verification problem illustrated in~\rref{fig:barriers} (left).
The task is to generate an invariant showing that ODE solutions starting within the initial set $\Init$ (in green) do not enter the unsafe set $\Unsafe$ (in red).
A candidate continuous invariant $p \leq 0$ (shown in blue in~\rref{fig:barriers}, left) is found using numerical barrier certificate generation techniques.
\begin{figure}[htp]
\centering
\includegraphics[width=0.65\columnwidth]{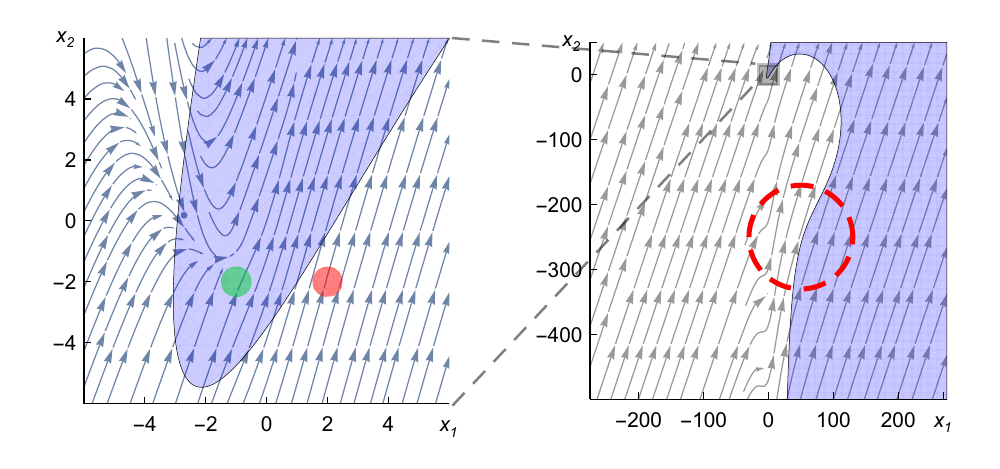}
\caption{(Left) A candidate invariant generated using numerical barrier certificates (in blue) for the safety verification problem of showing that solutions from the green initial state never reach the red unsafe states.
(Right) A zoomed out view of the safety verification problem, showing that the candidate invariant is, in fact, \emph{not} an invariant of the ODE because some states can exit the invariant (highlighted with a dashed red circle).}
\label{fig:barriers}
\end{figure}

The (exponential-type) barrier certificate $p$ is generated from a polynomial template $p_{\textcolor{red}{\vec{a}},d}$ of degree $d$ over variables $x,y$, by solving (and then substituting) for appropriate concrete values of the template coefficients $\textcolor{red}{\vec{a}}$.
For clarity below, the notation $p_{\textcolor{red}{\vec{a}},d}$ is used in steps where the generation algorithm produces constraints on the coefficients $\textcolor{red}{\vec{a}}$, while $p$ always refers to the final, generated barrier certificate.
Logically, it suffices to find real values for $\textcolor{red}{\vec{a}}$ so that the following formulas are simultaneously valid:
\begin{align}
& \Init \limply p_{\textcolor{red}{\vec{a}},d} \leq 0 \label{eq:barrier1} \,,\\
& \Unsafe \limply p_{\textcolor{red}{\vec{a}},d} > 0 \label{eq:barrier2} \,,\\
	& (p_{\textcolor{red}{\vec{a}},d})' \leq \lambda p_{\textcolor{red}{\vec{a}},d} \label{eq:barrier3}\enspace.
\end{align}

Constraints~\rref{eq:barrier1} and~\rref{eq:barrier2} ensure that the generated barrier separates the initial set from the unsafe set, e.g., in~\rref{fig:barriers} (left) the green initial region is wholly contained in the blue candidate invariant region $p \leq 0$, while the red unsafe region lies entirely outside.
Constraint~\rref{eq:barrier3} ensures that the sub-level set $p\leq 0$ is a continuous invariant, intuitively, the vector field points ``inwards'' along the boundary of $p \leq 0$ (blue region in~\rref{fig:barriers}), so it is impossible to flow from within $p \leq 0$ to $p > 0$.
A more general version of these constraints, and a soundness proof, is available elsewhere~\cite{DBLP:conf/cav/KongHSHG13}.

\end{example}

Sum-of-squares (SOS) programming~\cite{DBLP:journals/corr/PapachristodoulouAVPSP13} provides a tractable way of solving for the coefficients $\textcolor{red}{\vec{a}}$.
Suppose that $\Init,\Unsafe$ are described with polynomial inequalities $\Init \equiv \bigwedge_{i=1}^a I_i \geq 0$, $\Unsafe \equiv \bigwedge_{i=1}^b U_i \geq 0$.
Inequalities~\rref{eq:barrier1}--\rref{eq:barrier3} are respectively implied by the following SOS inequalities, where $\varepsilon > 0$ is a small positive constant and $\sigma_{I_i}$, $\sigma_{U_i}$ are template SOS polynomials~\cite{DBLP:journals/corr/PapachristodoulouAVPSP13}:
\begin{align}
&-p_{\textcolor{red}{\vec{a}},d} - \sum_{i=1}^a \sigma_{I_i} I_i \geq 0 \label{eq:barrier4} \,,\\
& p_{\textcolor{red}{\vec{a}},d} - \sum_{i=1}^b \sigma_{U_i} U_i - \varepsilon \geq 0  \label{eq:barrier5}\,,\\
	& \lambda p_{\textcolor{red}{\vec{a}},d} - (p_{\textcolor{red}{\vec{a}},d})' \geq 0 \label{eq:barrier6} \enspace.
\end{align}

Sum-of-squares solvers, such as SOSTOOLS~\cite{DBLP:journals/corr/PapachristodoulouAVPSP13}, witness the inequalities~\rref{eq:barrier4}--\rref{eq:barrier6} by finding an SOS representation for their left-hand side.
For example, a set of polynomials $g_1,\dots,g_n$ satisfying the polynomial identity $-p_{\textcolor{red}{\vec{a}},d} - \sum_{i=1}^a \sigma_{I_i} I_i  = \sum_{i=1}^n g_i^2$ proves~\rref{eq:barrier4} because the RHS of this inequality is a sum-of-squares, which is non-negative.
These polynomial identities are found efficiently by semidefinite programming~\cite{ParriloPhD}, which is also where \emph{numerical} solvers are used.
In practice, Pegasus loops through a range of values for the parameters $d, \lambda,\varepsilon$ as well as the degrees of the SOS polynomials $\sigma_{I_i}$, $\sigma_{U_i}$ and attempts to solve these constraints for each concrete choice of parameters.

While efficient, the use of numerical solvers has its drawbacks, e.g. because the generated coefficients $\textcolor{red}{\vec{a}}$ need not truly satisfy all the required constraints.
This is why Pegasus (and \KeYmaeraX) treats the generated barrier certificate $p$ only as a \emph{candidate} invariant and performs additional arithmetical checks to ensure that the constraints are truly met.
As a cautionary example,~\rref{fig:barriers} (left) rather misleadingly suggests that $p \leq 0$ is an invariant within its small plot domain.
Indeed, \rref{fig:barriers} (right) is a zoomed out version of the same plot which shows that the constraint~\rref{eq:barrier3} fails to hold for larger values of $x,y$.

Linear programming (LP) was employed as an alternative to sum-of-squares programming by Sankaranarayanan et al.~\cite{DBLP:conf/nolcos/Sankaranarayanan0A13} to generate Lyapunov functions, and later applied by Yang et al.~\cite{DBLP:conf/fm/YangHCL016} to similarly generate barrier certificates. The main idea behind this approach is to employ a \emph{linear relaxation}, whereby non-negativity of a polynomial $p$ is witnessed, subject to non-negativity of (basis) polynomials $p_1,p_2,\dots,p_k$, i.e.\ \mbox{$p_1\geq 0 \land p_2\geq 0 \land \dots \land p_k\geq 0 \limply p \geq 0$} is reduced to the existence of non-negative Lagrangian multipliers $\lambda_1,\lambda_2, \dots, \lambda_k$ such that $\lambda_1 p_1 + \lambda_2 p_2 + \cdots + \lambda_k p_k = p$.

In cases when the evolution domain constraint $Q$ is described by a conjunction of polynomial inequalities $Q \equiv q_1 \geq 0 \land \dots \land q_l \geq 0$ (e.g. in the case of hyperboxes or polyhedra), one may form all products $p_i=q_1^{\alpha_{1i}}\cdots q_l^{\alpha_{li}}$ up to some maximum total degree and use them to solve the linear relaxation for $p_1 \geq 0 \land \dots \land p_k \geq 0 \limply p_{\textcolor{red}{\vec{a}},d}\geq 0$ using linear programming, obtaining a polynomial which is non-negative on $Q$. The conditions for barrier certificates are encoded in an obvious way.

\begin{dangerous}
	In using convex optimization methods to search for barrier certificates, one is not concerned with optimizing the value of any particular objective function (the zero function suffices); one is rather interested in finding a feasible solution to a set of constraints.
  For LP, it is possible to use an SMT solver which supports the theory of linear real arithmetic (\texttt{LRA}, e.g., as supported by Z3) to search for \emph{models} of formulas describing the constraints to obtain instantiations of the parameter variables in the template; however, in our experience, implementations of linear programming solvers (especially employing interior point algorithms) in Mathematica and MATLAB offer considerably better performance compared to Z3 (which implements the Dual Simplex algorithm~\cite{DBLP:conf/cav/DutertreM06}).
\end{dangerous}

\section{Strategies for Invariant Generation}%
\label{sec:saturation}
The implementation of primitive invariant generation methods from~\rref{sec:pegasus} in a single framework is a significant undertaking in itself.
The overall goal behind Pegasus, however, is to enable these heterogeneous methods to be effectively deployed and fruitfully combined into \emph{strategies} for invariant generation that are tailored to specific classes of verification problems.
Different invariant generation strategies are invoked in Pegasus, depending on the classification of the input problem it receives from the problem \emph{classifier}.
In this section, and for the evaluation, we focus on the most challenging and general class of \emph{non-linear} systems in which no further structure is known or assumed beyond the fact that the right-hand sides of the ODEs are polynomials.

\subsection{Differential Saturation}
\label{subsec:diffsaturation}

The main invariant generation strategy Pegasus uses for general non-linear systems is based on a \emph{differential saturation} procedure~\cite{DBLP:journals/fmsd/PlatzerC09}.
Briefly, the procedure loops through a prescribed \emph{sequence} of invariant generation methods and \emph{successively} attempts to strengthen the evolution domain constraint using invariants found by those methods until the desired safety condition is proved.%
\footnote{Pegasus partitions problems into subsystems according to variable dependencies in their differential equations~\cite{DBLP:journals/fmsd/PlatzerC09}.
For $x_1'=x_1, x_2'= x_1 + x_2$, for example, Pegasus first searches for invariants involving only $x_1$, before searching for those involving both $x_1$ and $x_2$.}
Notably, this loop allows Pegasus to exploit the strengths of different invariant generation methods, even if it is \emph{a priori} unclear whether one is better than the other.
The precise sequencing of invariant generation methods is also important in this strategy to avoid redundancy.
Pegasus orders the methods by computational efficiency, e.g. it first searches for first integrals, followed by Darboux polynomials and barrier certificates.
This sequencing allows slower methods to exploit invariants that are quickly generated by earlier methods.

\begin{example}
\label{ex:diffsat}
\begingroup

The synergy between individual methods exploited by differential saturation is illustrated in \rref{fig:diffsat} for an example from our benchmarks.
\begin{figure}[t]
\includegraphics[width=\columnwidth]{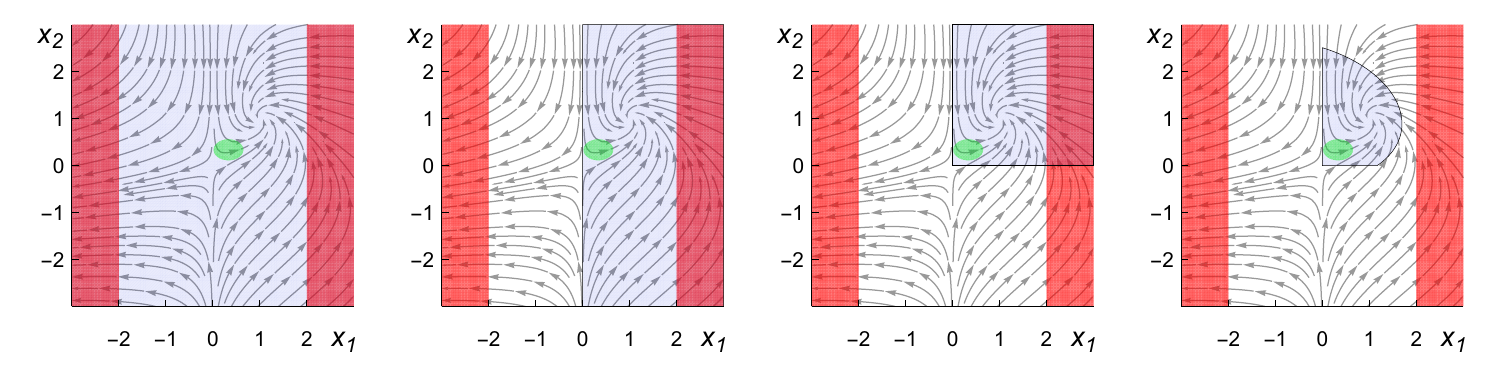}
\begin{subfigure}[t]{0.25\textwidth}
	\caption{No constraint\label{diffsatA}}
\end{subfigure}
\begin{subfigure}[t]{0.25\textwidth}
	\caption{$x_1 > 0$\label{diffsatB}}
\end{subfigure}
\begin{subfigure}[t]{0.25\textwidth}
	\caption{$x_1 > 0 \land x_2 > 0$\label{diffsatC}}
\end{subfigure}
\begin{subfigure}[t]{0.23\textwidth}
	\caption{$x_1 > 0 \land x_2 > 0 \land p < 0$\label{diffsatD}}
\end{subfigure}
\caption{Invariant synthesis using the differential saturation loop in Pegasus. The domain under consideration at each step is shaded in blue and annotated below each plot, with the polynomial $p = \frac{3}{8}x_1 + \frac{23}{56} x_1^2 - \frac{123}{56} x_2 + \frac{3}{14}x_1x_2 + \frac{29}{28} x_2^2 - 1$}
\label{fig:diffsat}
\end{figure}

	Initially (leftmost plot Fig.~\ref{diffsatA}), the entire plane (in blue) is under consideration and Pegasus wants to show the safety property that trajectories from the initial states (in green) never reach the unsafe states (in red).
	In the second plot (Fig.~\ref{diffsatB}), Pegasus confines its search to the region $x_1>0$ using the generated Darboux polynomial $x_1$.
	In the third plot (Fig.~\ref{diffsatC}), using $x_1>0$, qualitative analysis finds the invariant $x_2 > 0$ (whose invariance depends on $x_1>0$) which further confines the evolution domain constraint.
	Finally (rightmost plot Fig.~\ref{diffsatD}), Pegasus finds a barrier certificate (of polynomial degree 2) that suffices to show the
safety property within the strengthened evolution domain constraint (which, by construction, is invariant).
The final invariant region contains several sharp corners and thus \emph{cannot} be directly obtained as the sub-level set of a single polynomial barrier certificate.
Instead, it incorporates a conjunction of invariants discovered earlier by other means.
\endgroup
\end{example}

\begin{remark}
Pegasus extracts \emph{proof hints} from the internal reasoning sequence used in its differential saturation strategy, e.g., it tracks the order of construction of the invariants $x_1 > 0, x_2 > 0, \dots$ from~\rref{ex:diffsat} and how they were individually proved.
These hints are useful for deductive tools like \KeYmaeraX because they can be used to guide its proofs for the generated invariants in a corresponding, step-by-step manner, with the most appropriate verification technique for the invariant used at each step.
\end{remark}

Given an input safety verification problem, it is \emph{a priori} unknown which of the invariant generation methods used for differential saturation would succeed; and even for those that do succeed, it is difficult to predict the precise duration required.
The overall strategy in Pegasus imposes carefully balanced timeouts, where each method called by differential saturation attempts to:
\begin{itemize}
\item detect their applicability efficiently to conserve time budgets for other methods if they are not applicable,
\item keep track of intermediate results and report partial results (if applicable) when their individual timeouts are hit,
\item efficiently check when they are done.
\end{itemize}

Pegasus uses configuration parameters to adjust timeouts and method behavior, e.g., maximum degree of barrier certificate templates.
In addition, Pegasus supports configuration of the overall strategy behavior in terms of combining method results, how it handles method timeouts, and how it detects when the methods succeeded.
In the current implementation, and in~\rref{sec:evaluation}, we explore the following strategy configuration options:
\begin{enumerate}
\renewcommand{\labelenumi}{\textbf{\theenumi}}
\renewcommand{\theenumi}{\textbf{C\arabic{enumi}}}
\item \label{itm:c1} Auto-Reduction: whether or not to filter redundant invariants when combining results
\item \label{itm:c2} Heuristic Search: whether or not to apply qualitative analysis and other heuristic search methods
\item \label{itm:c3} Budget Redistribution: strict method timeouts or redistribution of unused time budget to later methods
\item \label{itm:c4} Subsystem Splitting: whether or not to analyze subsystems separately
\end{enumerate}

Option \ref{itm:c1} allows Pegasus to find invariants of lower descriptive complexity, which may be more insightful for users and easier to prove in \KeYmaeraX.
Options \ref{itm:c2}--\ref{itm:c4} allow expert users finer control over how Pegasus searches for invariants.
For example, \ref{itm:c4} is useful when the input problem is known to consist of many subsystems of ODEs~\cite{DBLP:journals/fmsd/PlatzerC09} that can be tackled separately.
The trade-off between these options is qualitatively evaluated in~\rref{sec:evaluation}.

\subsection{Differential Divide-and-Conquer}

The differential saturation strategy uses a melting pot of primitive invariant generation methods without (directly) adding more logical or mathematical considerations.
The \emph{differential divide-and-conquer} (DDC) proof rule~\cite{DBLP:conf/vmcai/SogokonGJP16} is an example logical technique that also fits well into the Pegasus framework.

Briefly, the rule says that if $p=0$ is an invariant for both the forwards ODE $\D{\vec{x}} = f(\vec{x})$ \emph{and} the backwards ODE $ \D{\vec{x}} = -f(\vec{x})$, then the state space partitions into three invariant subspaces $p < 0$, $p=0$, $p > 0$, and it suffices to consider the invariant generation sub-problems (restricted to each subspace) separately.
All Darboux polynomials $p$ (\rref{subsec:darboux}) meet the forwards-and-backwards invariance criteria and can be used to partition the state space.
Indeed, this DDC strategy is already implicitly used in the invariant generation method for rational first integrals in~\rref{subsec:ratFI}, which partitions the state space using Darboux polynomials, and then generates rational first integrals on the resulting sub-problems.
Pegasus generalizes this by looking for invariants on each sub-problem instead, i.e., by replacing steps~\ref{step4ratfi} and~\ref{step5ratfi} from the method described in~\rref{subsec:darboux} as follows:

\begin{enumerate}
	\item[4*] For each unpruned \emph{conflict cell} $S$, define a new invariant generation \emph{sub-problem}, with the original evolution domain constraint $Q$ restricted to $Q \land S$.
  \item[5*] Call the differential saturation strategy (\rref{subsec:diffsaturation}) to find an invariant on all newly generated sub-problems.
\end{enumerate}

\begin{example}
\begingroup

The differential divide-and-conquer strategy is illustrated in \rref{fig:diffdc} for a tweaked~\rref{ex:diffsat} with larger initial set and smaller unsafe set.
\begin{figure}
\includegraphics[width=\columnwidth]{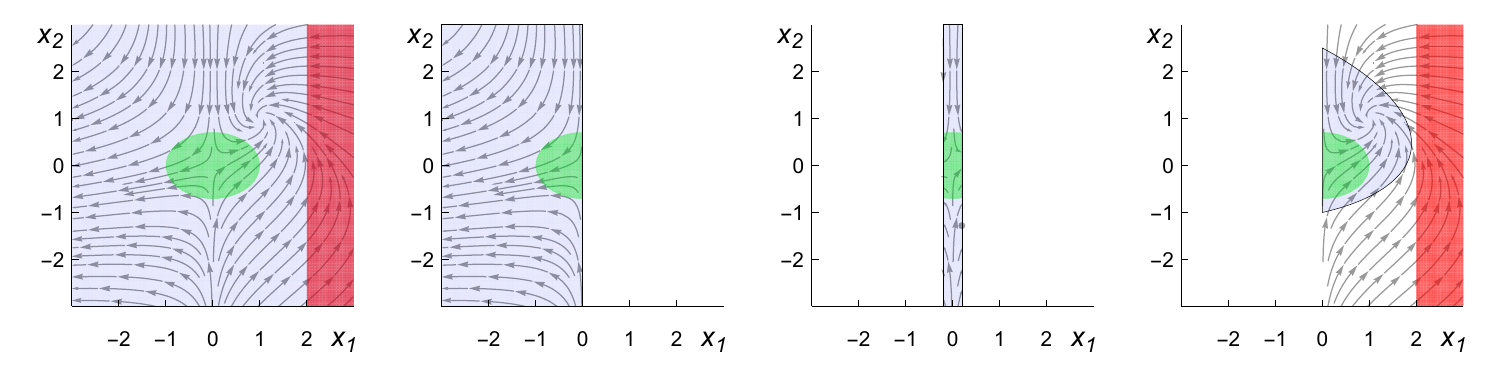}
\begin{subfigure}[t]{0.25\textwidth}
	\caption{No constraint\label{ddcA}}
\end{subfigure}
\begin{subfigure}[t]{0.25\textwidth}
	\caption{$x_1 < 0$\label{ddcB}}
\end{subfigure}
\begin{subfigure}[t]{0.25\textwidth}
	\caption{$x_1 = 0$ (enlarged)\label{ddcC}}
\end{subfigure}
\begin{subfigure}[t]{0.23\textwidth}
	\caption{$x_1 > 0 \land p < 0$\label{ddcD}}
\end{subfigure}
\caption{Invariant synthesis using differential divide-and-conquer in Pegasus.  The domain under consideration at each step is shaded in blue and annotated below each plot, with the polynomial $p = \frac{11}{25}x_1 + \frac{7}{100} x_1^2 - \frac{3}{5} x_2 + \frac{3}{25}x_1x_2 + \frac{2}{5} x_2^2 - 1$}
\label{fig:diffdc}
\end{figure}

	As before, initially (leftmost plot Fig~\ref{ddcA}), the entire plane (in blue) is under consideration and Pegasus wants to show the safety property that trajectories from the initial states (in green) never reach the unsafe states (in red).
	Pegasus partitions the problem into three sub-problems, shown in the subsequent plots, using the Darboux polynomial $x_1$; in those plots, only the part of the plane relevant to each sub-problem is drawn.In the third plot (Fig.~\ref{ddcC}, the evolution domain constraint $x_1=0$ is slightly (but soundly) enlarged to $-0.2\leq x_1 \leq 0.2$ for visibility in the illustration as it would otherwise be an infinitesimal line.
	In the second (evolution domain constraint $x_1<0$, Fig.~\ref{ddcB}) and third (evolution domain constraint $x_1=0$, enlarged in Fig.~\ref{ddcC}) plots, the sub-problems are proved trivially because they contain no unsafe states.
	In the rightmost plot (Fig.~\ref{ddcD}, evolution domain constraint $x_1 > 0$), Pegasus finds a barrier certificate (in blue) that solves the sub-problem.
\endgroup
\end{example}

\section{Evaluation}
\label{sec:evaluation}

This section presents a qualitative evaluation of the invariant generation capabilities of Pegasus and its interaction with the ODE proving tactics of \KeYmaeraX.
The insights obtained from these benchmarks provide useful default configuration options for Pegasus, e.g., those described in~\rref{sec:saturation}.

\subsection{Benchmark Suite}

The benchmark suite consists of 150 continuous safety verification problems, with 90 earlier problems~\cite{DBLP:conf/fm/SogokonMTCP19} and 60 new ones, all drawn from the literature~\cite{DBLP:conf/icalp/AlmagorKO020,DBLP:conf/cdc/SassiGS14,DBLP:journals/jsc/DaiGXZ17,DBLP:journals/automatica/DjaballahCKB17,ferragut2010,DBLP:conf/tacas/GhorbalP14,gorbuzov2012,DBLP:conf/cav/GulwaniT08,DBLP:conf/adhs/ImmlerA0FFKLMTZ18,DBLP:conf/hybrid/KapinskiDSA14,DBLP:conf/hybrid/KongBSJH17,DBLP:conf/emsoft/LiuZZ11,jaume2002,papachristodoulou2002,DBLP:conf/hybrid/Rodriguez-CarbonellT05,DBLP:conf/hybrid/Sankaranarayanan10,DBLP:conf/cpsweek/SogokonGJ16,DBLP:conf/fm/SogokonGTP18,DBLP:conf/fm/YangHCL016,yang2020,DBLP:journals/jacic/ZakiDTB09}.
Some are drawn from papers that present and discuss properties of a system of ODEs without explicitly providing initial and safe conditions; in such cases, we design our own initial and safe sets based on the provided discussion.

The suite consists of problems involving linear, affine, multi-affine, or (non-linear) polynomial ODEs over a range of dimensions: 71 two-dimensional systems, 30 three-dimensional systems, 35 higher-dimensional (${\geq} 4$, ${\leq}16$) systems, and 14 \emph{product systems} that were formed by randomly combining pairs of two- and three-dimensional systems, see~\rref{fig:classification} (a), (b).
The problems have a range of topological and logical structures to test the applicability of various invariant generation methods.
A summary of the topological structure of the problems is shown in~\rref{fig:classification} (c); the sets involved are either topologically bounded or unbounded (or None, when there is no evolution domain constraint), and either topologically open or closed (or neither).
A summary of the logical structure of the problems is shown in~\rref{fig:classification} (d); the formulas involved are either described algebraically by an equation, or by an atomic inequality, or, more generally, by a semi-algebraic formula involving conjunctions and disjunctions of equations and inequalities.
The experiment was run on commodity hardware.\footnote{\label{hardware} MacBook Pro 2019 with 2.6GHz Intel Core i7 (model 9750H) and 32GB memory (2667MHz DDR4 SDRAM), Mathematica 12.1 and MATLAB 2019b with SOSTOOLS 3.03.}

\begin{figure}[h]
\begin{subfigure}[t]{0.45\textwidth}
\begin{tikzpicture}
\pie[sum=auto,radius=1.5,after number=]{71/2-dim.,30/3-dim.,35/Higher-dim.,14/Product};
\end{tikzpicture}
\caption{Differential Equation Dimension}
\end{subfigure}\quad
\begin{subfigure}[t]{0.45\textwidth}
\begin{tikzpicture}
\pie[sum=auto,radius=1.5,after number=,text=label,/tikz/every label/.style={align=center}]{17/Linear,6/Affine,26/Multi-affine,5/Homogeneous,96/Polynomial};
\end{tikzpicture}
\caption{Differential Equation Class}
\end{subfigure}\\~\\

\begin{subfigure}[b]{\textwidth}
\centering
\begin{tabularx}{\columnwidth}{@{}X@{\hspace{-1em}}rrrrrrr@{}}
\toprule
                        & \multicolumn{3}{c}{Bounded}                                                         & \multicolumn{3}{c}{Unbounded}                                                       & \multicolumn{1}{l}{}    \\
\cmidrule(lr){2-4}\cmidrule(lr){5-7}
Topology                & \multicolumn{1}{l}{Open} & \multicolumn{1}{l}{Closed} & \multicolumn{1}{l}{Neither} & \multicolumn{1}{l}{Open} & \multicolumn{1}{l}{Closed} & \multicolumn{1}{l}{Neither} & \multicolumn{1}{l}{None} \\ \midrule
Initial Set (Pre.)      & 15                       & 76                         & 13                           & 20                       & 16                         & 10                           & -                         \\
Unsafe Set (Neg. Post.) & 1                        & 49                         & 2                            & 26                       & 57                         & 15                           & -                         \\
Evolution Domain        & 0                        & 26                         & 0                            & 3                        & 10                         & 0                            & 111                       \\ \bottomrule
\end{tabularx}
\caption{Problem topology}
\end{subfigure}\\~\\~\\
\begin{subfigure}[b]{\textwidth}
\centering
\begin{tabular}{@{}lrrrr@{}}
\toprule
Logical Structure   &  \multicolumn{1}{l}{Algebraic} &  \multicolumn{1}{l}{Atomic Inequality} &  \multicolumn{1}{l}{Semi-algebraic} & \multicolumn{1}{l}{None} \\ \midrule
Precondition        & 29        & 44                & 77                                                             & -                        \\
Postcondition       & 5         & 74                & 71                                                             & -                        \\
Evolution Domain    & 1         & 5                 & 33                                                             & 111                      \\ \bottomrule
\end{tabular}
\caption{Problem logical structure}
\end{subfigure}

\vspace{1ex}
\caption{Benchmark suite classification among 150 benchmarks}
\label{fig:classification}
\end{figure}
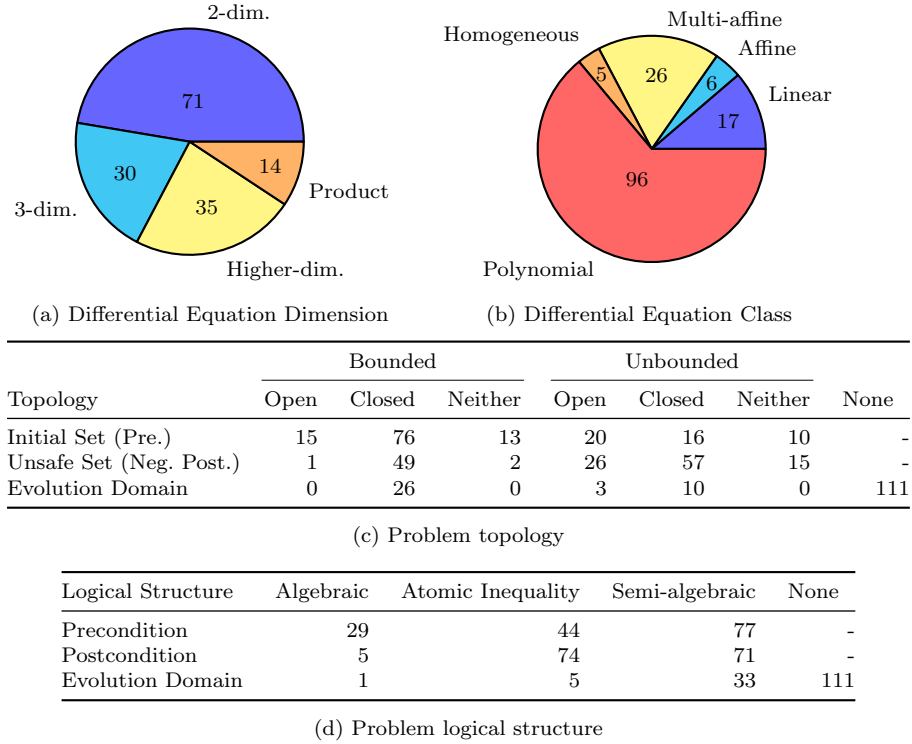

\subsection{Differential Saturation Performance}

We analyze the differential saturation strategy compared to each invariant generation method in isolation, measuring the duration of invariant generation, duration of checking the generated invariants, and the total proof duration.
We analyze the effect of exposing proof hints with the generated invariants, and the effect of strategy configuration options \ref{itm:c1}--\ref{itm:c4} from~\rref{sec:saturation}.

\subsubsection{Differential Saturation versus Individual Generation Methods}
The results comparing differential saturation against individual methods for each benchmark problem are shown in \rref{fig:strategycomparison}.
Several experimental insights can be drawn from these results:
\begin{enumerate*}[label=(\roman*)]
\item different invariant generation methods generally solve different subsets of the problems,
\item invariant generation almost always dominates total proof duration although invariant checking becomes more expensive as problem dimension increases,
\item when multiple methods solve a problem, qualitative analysis and first integrals are often quickest, followed by Darboux polynomials and then barrier certificates,
\item the differential saturation strategy effectively combines invariant generation methods; it solves $16$ additional problems (of which $7$ are product systems) that no individual method solves by itself.
Differential saturation is especially effective on product systems because each part of the product may be only solvable using a specific method.
\item Finally, the performance of Pegasus (with default configuration) has remained relatively stable compared to its earlier version~\cite{DBLP:conf/fm/SogokonMTCP19}.
\end{enumerate*}

\begin{figure}[h]
\centering
\begin{subfigure}[b]{\columnwidth}
\includegraphics[width=\columnwidth]{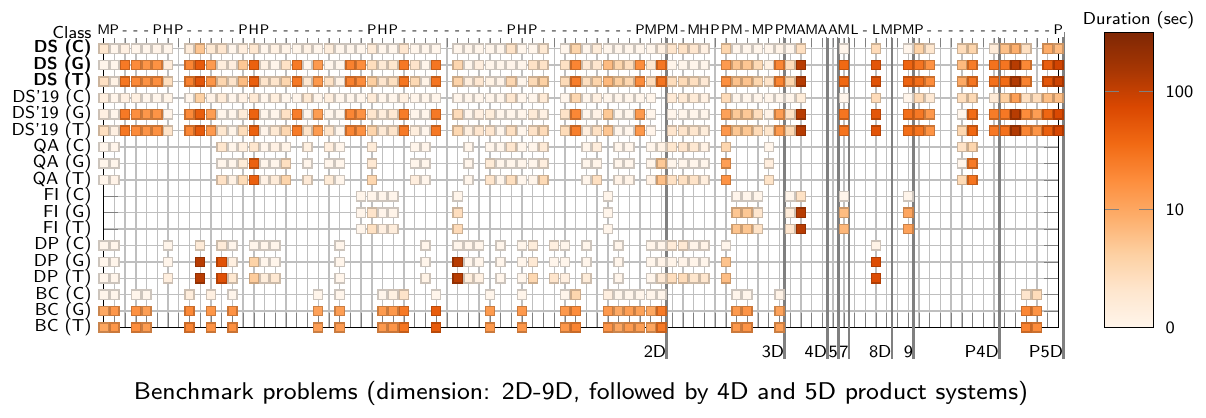}
\caption{90 benchmark problems from FM 2019 conference version \cite{DBLP:conf/fm/SogokonMTCP19}\vspace{\baselineskip}}
\end{subfigure}
\begin{subfigure}[t]{.2\columnwidth}
\begin{tikzpicture}
\end{tikzpicture}
\end{subfigure}
\begin{subfigure}[t]{.8\columnwidth}
\includegraphics[width=\columnwidth]{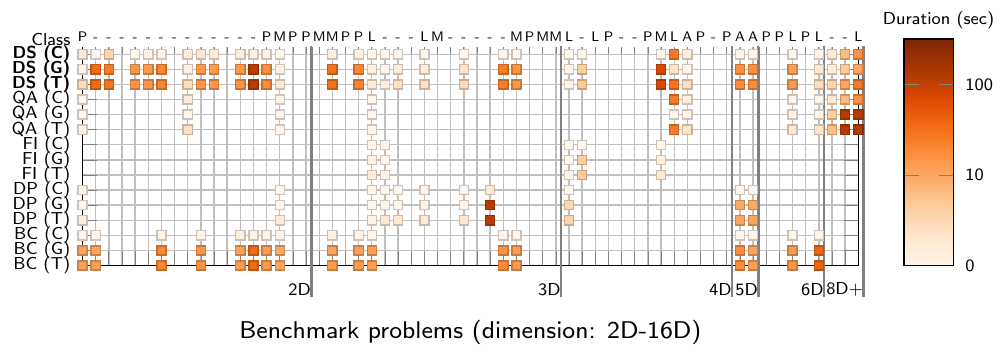}
\caption{60 additional benchmark problems\vspace{\baselineskip}}
\end{subfigure}
\caption{Comparison of invariant generation methods. Each column represents one benchmark problem and the color encodes duration (lighter is faster). Empty columns are unsolved. Legend: the combined Differential Saturation (DS) strategy against Qualitative Analysis (QA), First Integrals (FI), Darboux Polynomials (DP), and Barrier Certificates (BC), on total proof duration (T), generation duration (G), and checking duration (C).
Results for the earlier implementation~\cite{DBLP:conf/fm/SogokonMTCP19} (with new hardware, see Footnote~\ref{hardware}) are also shown for comparison (DS'19).
The ODE classification for each problem is annotated at the top: homogeneous polynomial (H), polynomial (P), linear (L), affine (A), multi-affine (M), dashes indicate same class as the enclosing labels.}
\label{fig:strategycomparison}
\end{figure}

\begin{figure}[H]
{\setlength{\belowcaptionskip}{0pt}
{\captionsetup[subfigure]{oneside,margin={0.8cm,0cm}}
\begin{subfigure}[b]{.36\textwidth}
\begin{tikzpicture}
\begin{axis}[
        font=\scriptsize\sffamily,
scale only axis,
width=0.75\columnwidth,
    ylabel={\quad Cumulative time (sec)},
xlabel={Problems},
ylabel near ticks,
xlabel style={at={(0,0.12)},anchor=east},
xtick={0,25,50,75,100},
ytick={0.1,1,10,100,1000,10000},
yticklabels={0.1,1,10,100,$10^3$},
ymin=0.1,
ymax=5000,
grid=major,
ymode=log,
legend columns=5,
legend entries={Diff. Sat., Barrier, Darboux, First Integrals, Qualitative,Full (no hints)},
legend style={at={(0.2,1.1)},anchor=south west}
]
    \addplot+[mark=none, thick, color=lsgreen,y domain=0:1000]                    table [col sep=comma,skip first n=2,y expr={\thisrowno{13}/1000},x index=12] {benchmarks_fmsd/20200905/all_proofhints_invgen_full_ranks.csv};
    \addplot+[mark=none, dashed,y domain=0:1000]           table [col sep=comma,skip first n=2,y expr={\thisrowno{13}/1000},x index=12] {benchmarks_fmsd/20200905/all_proofhints_invgen_barrier_ranks.csv};
    \addplot+[mark=none, densely dotted,y domain=0:1000]    table [col sep=comma,skip first n=2,y expr={\thisrowno{13}/1000},x index=12] {benchmarks_fmsd/20200905/all_proofhints_invgen_dbx_ranks.csv};
    \addplot+[mark=none,y domain=0:1000]             table [col sep=comma,skip first n=2,y expr={\thisrowno{13}/1000},x index=12] {benchmarks_fmsd/20200905/all_proofhints_invgen_firstintegrals_ranks.csv};
    \addplot+[mark=none, densely dashdotted,y domain=0:1000]    table [col sep=comma,skip first n=2,y expr={\thisrowno{13}/1000},x index=12] {benchmarks_fmsd/20200905/all_proofhints_invgen_summands_ranks.csv};
\end{axis}
\end{tikzpicture}
\caption{\mbox{Total duration}}
\end{subfigure}}%
\quad
\begin{subfigure}[b]{.27\textwidth}
\begin{tikzpicture}
\begin{axis}[
        font=\scriptsize\sffamily,
scale only axis,
width=\columnwidth,
ylabel near ticks,
yticklabels={},
xtick={0,25,50,75,100},
ytick={0.1,1,10,100,1000,10000},
ymin=0.1,
ymax=5000,
grid=major,
ymode=log,
]
\addplot+[mark=none, color=lsgreen]                    table [col sep=comma,skip first n=2,y expr={\thisrowno{17}/1000},x index=12] {benchmarks_fmsd/20200905/all_proofhints_invgen_full_ranks.csv};
\addplot+[mark=none, dashed ]           table [col sep=comma,skip first n=2,y expr={\thisrowno{17}/1000},x index=12] {benchmarks_fmsd/20200905/all_proofhints_invgen_barrier_ranks.csv};
\addplot+[mark=none, densely dotted]    table [col sep=comma,skip first n=2,y expr={\thisrowno{17}/1000},x index=12] {benchmarks_fmsd/20200905/all_proofhints_invgen_dbx_ranks.csv};
\addplot+[mark=none, solid]             table [col sep=comma,skip first n=2,y expr={\thisrowno{17}/1000},x index=12] {benchmarks_fmsd/20200905/all_proofhints_invgen_firstintegrals_ranks.csv};
\addplot+[mark=none, densely dashdotted]table [col sep=comma,skip first n=2,y expr={\thisrowno{17}/1000},x index=12] {benchmarks_fmsd/20200905/all_proofhints_invgen_summands_ranks.csv};
\addplot+[mark=none, color=lsgreen, thin] 
table [col sep=comma,skip first n=2,y expr={\thisrowno{21}/1000},x index=12] {benchmarks_fmsd/20200905/all_proofhints_invgen_full_ranks.csv};
\node[circle,fill,inner sep=1pt] at (axis cs:100,121) {};
\node[label={270:{$\times 8$}},circle,fill,inner sep=1pt] at (axis cs:100,997) {};
\node[circle,fill,inner sep=1pt] at (axis cs:75,65) {};
\node[label={270:{$\times 6$}},circle,fill,inner sep=1pt] at (axis cs:75,373) {};
\node[label={270:{$\times 2$}},circle,fill,inner sep=1pt] at (axis cs:50,34) {};
\node[circle,fill,inner sep=1pt] at (axis cs:50,75) {};
\end{axis}
\end{tikzpicture}
\caption{Generation}
\end{subfigure}
\quad
\begin{subfigure}[b]{.27\textwidth}
\begin{tikzpicture}
\begin{axis}[
        font=\scriptsize\sffamily,
scale only axis,
width=\columnwidth,
ylabel near ticks,
yticklabels={},
xtick={0,25,50,75,100},
ytick={0.1,1,10,100,1000,10000},
ymin=0.1,
ymax=5000,
grid=major,
ymode=log,
]
\addplot+[mark=none, color=lsgreen]                    table [col sep=comma,skip first n=2,y expr={\thisrowno{21}/1000},x index=12] {benchmarks_fmsd/20200905/all_proofhints_invgen_full_ranks.csv};
\addplot+[mark=none, dashed ]           table [col sep=comma,skip first n=2,y expr={\thisrowno{21}/1000},x index=12] {benchmarks_fmsd/20200905/all_proofhints_invgen_barrier_ranks.csv};
\addplot+[mark=none, densely dotted]    table [col sep=comma,skip first n=2,y expr={\thisrowno{21}/1000},x index=12] {benchmarks_fmsd/20200905/all_proofhints_invgen_dbx_ranks.csv};
\addplot+[mark=none, solid]             table [col sep=comma,skip first n=2,y expr={\thisrowno{21}/1000},x index=12] {benchmarks_fmsd/20200905/all_proofhints_invgen_firstintegrals_ranks.csv};
\addplot+[mark=none, densely dashdotted]table [col sep=comma,skip first n=2,y expr={\thisrowno{21}/1000},x index=12] {benchmarks_fmsd/20200905/all_proofhints_invgen_summands_ranks.csv};
\addplot+[mark=none, color=lsgreen, thin] 
table [col sep=comma,skip first n=2,y expr={\thisrowno{17}/1000},x index=12] {benchmarks_fmsd/20200905/all_proofhints_invgen_full_ranks.csv};
\node[circle,fill,inner sep=1pt] at (axis cs:100,121) {};
\node[label={270:{$\times 8$}},circle,fill,inner sep=1pt] at (axis cs:100,997) {};
\node[circle,fill,inner sep=1pt] at (axis cs:75,65) {};
\node[label={270:{$\times 6$}},circle,fill,inner sep=1pt] at (axis cs:75,373) {};
\node[label={270:{$\times 2$}},circle,fill,inner sep=1pt] at (axis cs:50,34) {};
\node[circle,fill,inner sep=1pt] at (axis cs:50,75) {};
\end{axis}
\end{tikzpicture}
\caption{Checking }
\end{subfigure}}%
\caption{Cumulative logarithmic time (in seconds) taken to solve the fastest $n$ problems (more problems solved and flatter is better; accumulated generation and checking duration of Diff. Sat. compared at 50/75/100 fastest problems)}
\label{fig:accumulatedtimes}
\end{figure}
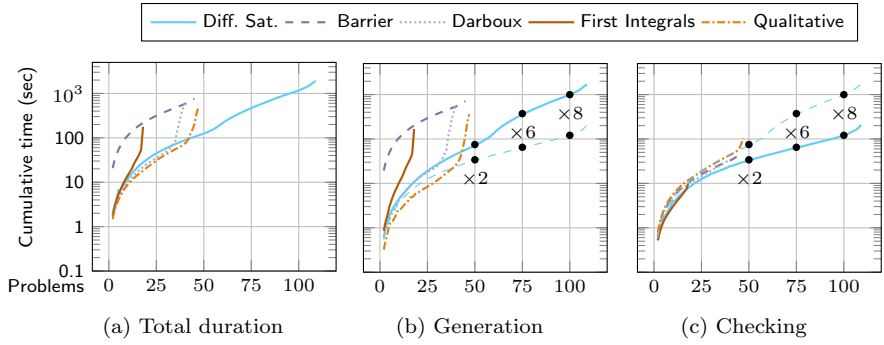

To evaluate the effectiveness of combining methods by differential saturation, \rref{fig:accumulatedtimes} plots the \emph{accumulated} duration for solving the fastest $n$ out of $150$ benchmark problems.
The main insights are:
\begin{enumerate*}[label=(\roman*)]
\item differential saturation solves the largest number of problems per accumulated time, i.e., despite sequentially executing invariant generation methods, it often succeeds in trying out the most efficient method first and fails fast when earlier methods are unsuitable; however, qualitative analysis (in isolation) generates some invariants faster when the heuristics it employs for guessing invariant candidates are successful,
\item cumulatively, invariant generation duration dominates invariant checking duration (note logarithmic scaling of the time axis in~\rref{fig:accumulatedtimes}); this effect is especially pronounced for barrier certificates, but can also be observed in all other methods when solving more expensive (harder) problems,
\item first integrals are least expensive to check when they solve problems,
\item qualitative analysis is less expensive for generation than other methods, but is most expensive for checking because the invariants it generates often have high descriptive complexity and may not have simple invariance justifications.

\end{enumerate*}

\subsubsection{Differential Saturation Configuration Options}

Next, we explore the effect of configuration options on the invariant generation and subsequent checking duration by disabling features of the differential saturation procedure.
Specifically, we executed differential saturation  with:
\begin{itemize}[label=DUMY,wide=0pt, leftmargin=*]
\item[\rref{itm:c1}\cancel{AR}] No Auto-Reduction, which is expected to speed up generation but may cause redundant cuts or unnecessarily complicated invariants.
\item[\rref{itm:c2}\cancel{HS}] No Heuristic Search, which is expected to produce more principled invariants and more specific proof hints but solve fewer problems.
\item[\rref{itm:c3}\cancel{BR}] No Budget Redistribution, which is expected to result in a more predictable generation duration but solve fewer problems.
\item[\rref{itm:c4}\cancel{SS}] No Subsystem Splitting, which is expected to result in faster performance on problems without clear subsystems, but solve fewer problems overall (e.g., the product problems should benefit from~\rref{itm:c4}).
\item[\cancel{PH}] No Proof Hints, which is expected to slow down invariant checking but have no effect on invariant generation.
\end{itemize}

Figure~\ref{fig:stratconfigcomp} shows the benefits and drawbacks of each configuration option on the suite of benchmark problems, while~\rref{fig:configaccumulatedtimes} summarizes the cumulative effect of configuration options. Since these configuration options are tuning parameters that offer fine-grained control over differential saturation for Pegasus, their cumulative effect over all 150 problems is small, see \rref{fig:configaccumulatedtimes}.

\begin{figure}[tb]
\begin{subfigure}[b]{\columnwidth}
\includegraphics[width=\columnwidth]{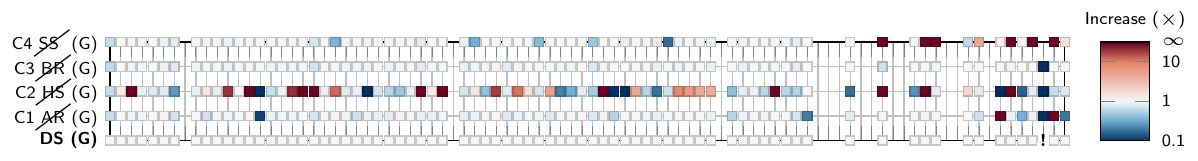}
\includegraphics[width=\columnwidth]{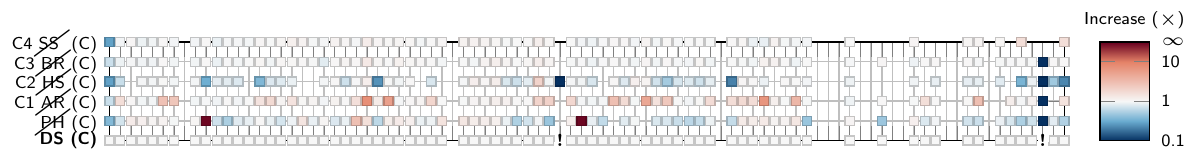}
\caption{Invariant generation (top) and checking (bottom) duration in multiples of differential saturation (90 benchmark problems from FM 2019 conference version \cite{DBLP:conf/fm/SogokonMTCP19})\vspace{\baselineskip}}
\label{fig:diffsatregreta}
\end{subfigure}
\begin{subfigure}[b]{\columnwidth}
\includegraphics[width=\columnwidth]{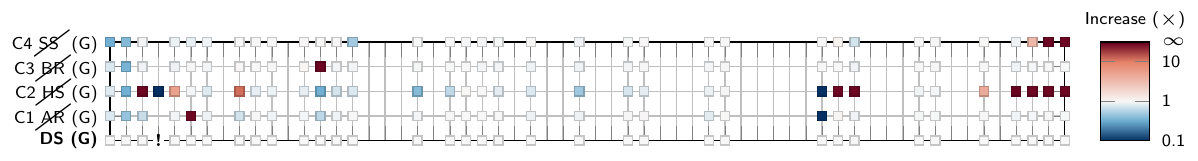}
\includegraphics[width=\columnwidth]{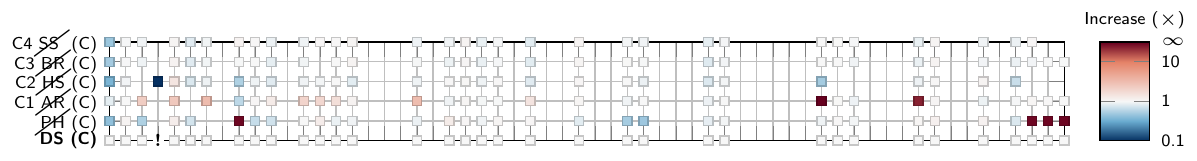}
\caption{Invariant generation (top) and checking (bottom) duration in multiples of differential saturation (60 additional benchmark problems\vspace{\baselineskip})}
\label{fig:diffsatregretb}
\end{subfigure}
\caption{Influence of configuration options: no Auto-Reduction (\rref{itm:c1}\cancel{AR}), no Heuristic Search (\rref{itm:c2}\cancel{HS}), no Budget Redistribution (\rref{itm:c3}\cancel{BR}), no Subsystem Splitting (\rref{itm:c4}\cancel{SS}), and no Proof Hints (\cancel{PH}). A \textbf{!} mark indicates that the default Differential Saturation (DS) configuration failed to generate or check that problem, while one (or more) of the other configuration options succeeded.} %
\label{fig:stratconfigcomp}
\end{figure}

Except for Heuristic Search (\rref{itm:c2}), disabling features results in similar (or slightly faster) generation duration for most problems, but at the expense of not solving others, see Figs.~\ref{fig:diffsatregreta} and~\ref{fig:diffsatregretb} (top).
On three particular problems, disabling features helped Pegasus to solve the problem within the given time budget.
Overall, the configuration options have little net effect on most problems but can make a difference on select problems:
\newcommand{\concl}{Conclusion:\xspace}%
\begin{itemize}
\item No Proof Hints (\cancel{PH}): Several problems check slightly faster \emph{without} following the proof hints, which indicates that \KeYmaeraX's checking procedure is sometimes able to find more efficient proofs than the hints.
However, there are also problems that check slightly slower and several problems that fail to check without proof hints.
\concl proof hints can be extremely helpful; they should be kept wherever possible, especially since they are inexpensive to produce in Pegasus. \KeYmaeraX could try its default checking procedure first and fallback to hints if the default fails.
\item No Auto-Reduction (\rref{itm:c1}\cancel{AR}): significant increase in proof checking duration on several examples, but decrease in generation duration on several examples as well.
\concl \rref{itm:c1} auto-reduction is useful for checking but at the expense of generation duration; it should be provided as an optional post-processing step for users interested in more succinct invariants.
\item No Heuristic Search (\rref{itm:c2}\cancel{HS}): variable severe impact (both positive and negative) on generation duration across examples, but fails to generate invariants for several examples.
However, checking duration is generally improved for principled invariants generated without heuristics.
Notably, two problems were successfully solved \emph{solely} by \rref{itm:c2}\cancel{HS} out of all other configuration options.
\concl \rref{itm:c2} should be a configurable option for users, but should typically be enabled when the ultimate goal is to solve a given problem and invariant generation time is not a significant constraint.
\item No Budget Redistribution (\rref{itm:c3}\cancel{BR}): minor impact on both generation and checking duration, except failing to solve one problem. \concl \rref{itm:c3} is not very impactful, but could be left enabled by default as a failsafe.
\item No Subsystem Splitting (\rref{itm:c4}\cancel{SS}): minor impact on both generation and checking duration for solved problems, but solves fewer problems (mostly product system and higher-dimensional problems).
\concl \rref{itm:c4} is a useful technique in invariant generation and should typically be enabled.
\end{itemize}

\begin{figure}[th]
{\setlength{\belowcaptionskip}{0pt}
{\captionsetup[subfigure]{oneside,margin={0.8cm,0cm}}
\begin{subfigure}[b]{.36\textwidth}
\begin{tikzpicture}
\begin{axis}[
        font=\scriptsize\sffamily,
scale only axis,
width=0.75\columnwidth,
    ylabel={\quad Cumulative time (sec)},
xlabel={Problems},
ylabel near ticks,
xlabel style={at={(0,0.12)},anchor=east},
xtick={0,25,50,75,100},
ytick={0.1,1,10,100,1000,10000},
yticklabels={0.1,1,10,100,$10^3$},
ymin=0.1,
ymax=5000,
grid=major,
ymode=log,
legend columns=6,
legend entries={Diff. Sat., \cancel{AR}, \cancel{HS}, \cancel{BR}, \cancel{SS}, \cancel{PH}},
legend style={at={(0.2,1.1)},anchor=south west}
]
    \addplot+[mark=none, thick, color=lsgreen,y domain=0:1000]                    table [col sep=comma,skip first n=2,y expr={\thisrowno{13}/1000},x index=12] {benchmarks_fmsd/20200905/all_proofhints_invgen_full_ranks.csv};
    \addplot+[mark=none, dashed,y domain=0:1000]           table [col sep=comma,skip first n=2,y expr={\thisrowno{13}/1000},x index=12] {benchmarks_fmsd/20200905/all_proofhints_invgen_full-nocut_ranks.csv};
    \addplot+[mark=none, densely dotted,y domain=0:1000]    table [col sep=comma,skip first n=2,y expr={\thisrowno{13}/1000},x index=12] {benchmarks_fmsd/20200905/all_proofhints_invgen_full-noheur_ranks.csv};
    \addplot+[mark=none,y domain=0:1000]             table [col sep=comma,skip first n=2,y expr={\thisrowno{13}/1000},x index=12] {benchmarks_fmsd/20200905/all_proofhints_invgen_full-stricttime_ranks.csv};
    \addplot+[mark=none, densely dashdotted,y domain=0:1000]    table [col sep=comma,skip first n=2,y expr={\thisrowno{13}/1000},x index=12] {benchmarks_fmsd/20200905/all_proofhints_invgen_full-nodep_ranks.csv};
        \addplot+[mark=none, color=lsgreen, loosely dashdotted,y domain=0:1000]    table [col sep=comma,skip first n=2,y expr={\thisrowno{13}/1000},x index=12] {benchmarks_fmsd/20200905/all_proofhints_invgen_full-nohints_ranks.csv};
\end{axis}
\end{tikzpicture}
\caption{\mbox{Total duration}}
\end{subfigure}}%
\quad
\begin{subfigure}[b]{.27\textwidth}
\begin{tikzpicture}
\begin{axis}[
        font=\scriptsize\sffamily,
scale only axis,
width=\columnwidth,
ylabel near ticks,
yticklabels={},
xtick={0,25,50,75,100},
ytick={0.1,1,10,100,1000,10000},
ymin=0.1,
ymax=5000,
grid=major,
ymode=log,
]
\addplot+[mark=none, color=lsgreen]                    table [col sep=comma,skip first n=2,y expr={\thisrowno{17}/1000},x index=12] {benchmarks_fmsd/20200905/all_proofhints_invgen_full_ranks.csv};
\addplot+[mark=none, dashed ]           table [col sep=comma,skip first n=2,y expr={\thisrowno{17}/1000},x index=12] {benchmarks_fmsd/20200905/all_proofhints_invgen_full-nocut_ranks.csv};
\addplot+[mark=none, densely dotted]    table [col sep=comma,skip first n=2,y expr={\thisrowno{17}/1000},x index=12] {benchmarks_fmsd/20200905/all_proofhints_invgen_full-noheur_ranks.csv};
\addplot+[mark=none, solid]             table [col sep=comma,skip first n=2,y expr={\thisrowno{17}/1000},x index=12] {benchmarks_fmsd/20200905/all_proofhints_invgen_full-stricttime_ranks.csv};
\addplot+[mark=none, densely dashdotted]table [col sep=comma,skip first n=2,y expr={\thisrowno{17}/1000},x index=12] {benchmarks_fmsd/20200905/all_proofhints_invgen_full-nodep_ranks.csv};
\addplot+[mark=none, color=lsgreen, loosely dashdotted]table [col sep=comma,skip first n=2,y expr={\thisrowno{17}/1000},x index=12] {benchmarks_fmsd/20200905/all_proofhints_invgen_full-nohints_ranks.csv};
\end{axis}
\end{tikzpicture}
\caption{Generation }
\end{subfigure}
\quad
\begin{subfigure}[b]{.27\textwidth}
\begin{tikzpicture}
\begin{axis}[
        font=\scriptsize\sffamily,
scale only axis,
width=\columnwidth,
ylabel near ticks,
yticklabels={},
xtick={0,25,50,75,100},
ytick={0.1,1,10,100,1000,10000},
ymin=0.1,
ymax=5000,
grid=major,
ymode=log,
]
\addplot+[mark=none, color=lsgreen]                    table [col sep=comma,skip first n=2,y expr={\thisrowno{21}/1000},x index=12] {benchmarks_fmsd/20200905/all_proofhints_invgen_full_ranks.csv};
\addplot+[mark=none, dashed ]           table [col sep=comma,skip first n=2,y expr={\thisrowno{21}/1000},x index=12] {benchmarks_fmsd/20200905/all_proofhints_invgen_full-nocut_ranks.csv};
\addplot+[mark=none, densely dotted]    table [col sep=comma,skip first n=2,y expr={\thisrowno{21}/1000},x index=12] {benchmarks_fmsd/20200905/all_proofhints_invgen_full-noheur_ranks.csv};
\addplot+[mark=none, solid]             table [col sep=comma,skip first n=2,y expr={\thisrowno{21}/1000},x index=12] {benchmarks_fmsd/20200905/all_proofhints_invgen_full-stricttime_ranks.csv};
\addplot+[mark=none, densely dashdotted]table [col sep=comma,skip first n=2,y expr={\thisrowno{21}/1000},x index=12] {benchmarks_fmsd/20200905/all_proofhints_invgen_full-nodep_ranks.csv};
\addplot+[mark=none, color=lsgreen, loosely dashdotted]table [col sep=comma,skip first n=2,y expr={\thisrowno{21}/1000},x index=12] {benchmarks_fmsd/20200905/all_proofhints_invgen_full-nohints_ranks.csv};
\end{axis}
\end{tikzpicture}
\caption{Checking}
\end{subfigure}}%
\caption{Configuration options: cumulative logarithmic time (in seconds) taken to solve the fastest $n$ problems (more problems solved and flatter is better)}
\label{fig:configaccumulatedtimes}
\end{figure}
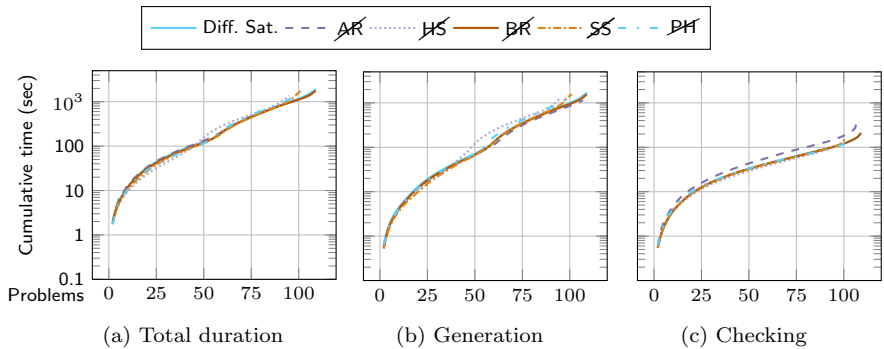

\section{Related Work}
\label{sec:related}
Techniques developed for \emph{qualitative simulation} have been applied to
prove temporal properties of continuous systems by Shults and
Kuipers~\cite{DBLP:journals/ai/ShultsK97}, as well as Loeser, Iwasaki and
Fikes~\cite{LoeserIwasakiFikes:1998}.
Zhao \cite{zhao1994extracting} developed a tool, MAPS, to automatically
identify significant features of dynamical systems, such as stability regions,
equilibria, and limit cycles. Since
our ultimate goal is sound invariant generation, we are less interested in a
full qualitative analysis of the state space.
In the verification community, discrete abstraction of hybrid systems was
studied by Alur \emph{et al}.~\cite{Alur2000}. The case of systems whose continuous
motion is governed by non-linear ODEs was studied in the work of Tiwari and
Khanna~\cite{Tiwari2008FMSD,TiwariKhanna2002HSCC}.
Tiwari studied reachability of linear systems~\cite{Tiwari2003},
using information from real eigenvectors and ideas from qualitative abstraction to generate invariants.
Zaki \emph{et al.}~\cite{DBLP:journals/jacic/ZakiDTB09} were the first to apply Darboux polynomials
to verification of continuous systems using discrete abstraction. Numerous works employ barrier
certificates for verification~\cite{DBLP:journals/jsc/DaiGXZ17,DBLP:conf/cav/KongHSHG13,DBLP:conf/hybrid/PrajnaJ04,DBLP:conf/fm/SogokonGTP18,DBLP:conf/fm/YangHCL016}.
Since we implement many of the above techniques as methods for invariant generation in Pegasus, our work draws heavily upon ideas developed previously in the
verification and hybrid systems communities.
Previous work~\cite{DBLP:conf/vmcai/SogokonGJP16} introduced a construction of exact
abstractions and applied rudimentary methods from qualitative
analysis to compute invariants; in certain ways, our present work also builds on this
experience, incorporating some of the techniques as special methods in a more general framework.
The coupling between \KeYmaeraX and Pegasus that we pursue is quite
distinct from the use of trusted oracles in the work of Wang \emph{et al.}~\cite{DBLP:conf/icfem/WangZZ15}
(for the HHL prover) and, notably, provides a \emph{sound} framework for reasoning with continuous
invariants that is significantly less exposed to soundness issues in external tools.

A \emph{complete} semi-algorithm for computing algebraic invariants (described by zero sets of polynomial
functions) for polynomial systems of ODEs was developed by Ghorbal and Platzer~\cite{DBLP:conf/tacas/GhorbalP14}.
An interesting development along very similar lines was also recently pursued
by Boreale~\cite{Boreale2020}, whose method makes use of the algebraic nature
of the precondition (initial set) in the verification problem in order to speed
up the algebraic invariant generation. Both of these (semi-)algorithms involve
enumeration of polynomial templates; the biggest practical difficulty stems
from the computational cost of minimizing the rank of symbolic matrices~\cite{DBLP:conf/tacas/GhorbalP14},
and computing the generators of \emph{real radical ideals}~\cite{Boreale2020}, both of which are difficult problems
with the latter having few algorithms with robust implementations currently in existence.\footnote{Although an \emph{incomplete} invariant generation procedure could still employ inexpensive ad-hoc methods to compute generators of real radical ideals; likewise, generators of (complex) radical ideals can be used instead in a sound but incomplete algebraic invariant generation algorithm~\cite[\S 5]{Boreale2020}.}
In the future, we hope to extend Pegasus with an implementation of these techniques.

\section{Outlook and Challenges}
\label{sec:outlook}
The improvements in continuous invariant generation have a significant impact
on the overall proof automation capabilities of \KeYmaeraX and serve to
increase overall system usability and improve user experience. Better proof
automation will certainly also be useful in future applications of provably
correct runtime monitoring frameworks, such as
ModelPlex~\cite{DBLP:journals/fmsd/MitschP16}, as well as frameworks for
generating verified controller executables, such as
VeriPhy~\cite{DBLP:conf/pldi/BohrerTMMP18}.
Some interesting directions for extending our work include implementation of
reachable set computation algorithms for all classes of problems where this is
possible. For instance, semi-algebraic reachable sets for
diagonalizable classes of linear systems with tame eigenvalues~\cite{gan2018reachability,DBLP:journals/jsc/LafferrierePY01}, as well as more generally~\cite{DBLP:conf/icalp/AlmagorKO020}.
The complexity of invariants obtained using these methods may not
always make them practical, but they would provide a valuable fallback when simpler invariants cannot be obtained using our currently
implemented methods.

A more pressing challenge lies in expanding the collection of safety verification
problems for continuous systems. While we have done our best to find compelling
examples from the literature, a larger corpus of problems would allow for a
more comprehensive empirical evaluation of invariant generation strategies and
could reveal interesting new insights that can suggest more effective strategies.

Correctness of decision procedures for real arithmetic is another important challenge.
For pragmatic reasons, \KeYmaeraX currently uses Mathematica's  implementation
of real quantifier elimination to check validity of first-order real arithmetic formulas.
Removing this reliance by efficiently building fully formal proofs of real
arithmetic formulas within \dL (e.g. through exhibiting appropriate
witnesses or using proof-producing procedures; see~\cite{DBLP:conf/cade/PlatzerQR09} 
for an overview) is an important task for the future.

Other important topics not addressed in this article concern \emph{stability} and
\emph{robustness} of continuous invariants~\cite{GoebelRobust,10.2307/j.ctvcm4hws,Khalil,DBLP:conf/tacas/KongGCC15}. 
These notions are important in ensuring that the generated invariants are
reflective of the real world, and are not merely by-products of mathematical
idealization.

\section{Conclusion}
\label{sec:conclusion}
Among verification practitioners, the amount of manual effort required for
formal verification of hybrid systems is one of the chief criticisms leveled
against the use of deductive verification tools.  Manually crafting continuous
invariants may require expertise and ingenuity, just like manually selecting
support function templates for reachability tools
\cite{DBLP:conf/cav/FrehseGDCRLRGDM11}, and presents a major practical hurdle
in the way of wider industrial adoption of this technology.
In this article, we describe our development of a system designed to help
overcome this hurdle by automating the discovery of continuous invariants.  To
our knowledge, this work represents the first large-scale effort in combining
continuous invariant generation methods into a single invariant generation
framework and making it possible to create more powerful invariant generation
strategies.  The approach we pursue is unique in its integration with a theorem
prover, which provides formal guarantees that the generated invariants are
indeed correct (in the form of \dL proofs, \emph{automatically}).  The results
we observe in our evaluation are highly encouraging and suggest that invariant
discovery can be improved considerably, opening many exciting avenues for
applications and extensions.

\begin{acknowledgements}
The authors would like to thank the anonymous reviewers for providing valuable
feedback and FM 2019 for the special issue invitation.
\end{acknowledgements}

\bibliographystyle{spmpsci}      %
\bibliography{root}   %

\end{document}